\documentclass[10pt]{amsart}

\usepackage[square,sort&compress]{natbib}
\usepackage[sans]{dsfont}
\usepackage{url,amsmath,amsopn}
\usepackage{amssymb}
\usepackage{a4wide}
\usepackage{amsthm}
\usepackage{stmaryrd}
\usepackage{color}
\usepackage{float}
\usepackage{pict2e}
\usepackage{bm}
\usepackage{subfigure}
\usepackage{dsfont}
\usepackage{float}
\usepackage{graphicx}
\usepackage{scalefnt}
\usepackage{color}
\usepackage{rotfloat}
\usepackage[pdftex,%
pagebackref=true,%
colorlinks=true,%
bookmarks=false,%
linkcolor=red,%
citecolor=black,%
urlcolor=black,%
pdfstartview=FitV,%
pdfpagelayout=TwoColumnRight]{hyperref}

\floatstyle{ruled}
\newfloat{algorithm}{tbp}{loa}[section]
\floatname{algorithm}{Algorithm}

\newcommand{\x}{\bm x}
\renewcommand{\t}{\bm t}
\newcommand{\z}{\bm z}
\newcommand{\y}{\bm y}
\newcommand{\nod}{\emptyset}
\newcommand{\R}{\mathds R}
\renewcommand{\(}{\left(}
\renewcommand{\)}{\right)}

\renewcommand{\u}{\mathbf{u}}
\newcommand{\f}{v}
\newcommand{\g}{g}
\newcommand{\ff}{\bm v}

\newcommand{\G}{G}
\newcommand{\F}{F}
\newcommand{\twu}{\tilde{u}}
\newcommand{\twbu}{\tilde{\bm u}}
\renewcommand{\th}{${}^{\text{th}}$ }

\theoremstyle{plain}
\newtheorem{lemma}{Lemma}
\newtheorem{proposition}{Proposition}
\newtheorem{corollary}{Corollary}
\theoremstyle{definition}
\newtheorem{definition}{Definition}
\theoremstyle{remark}
\newtheorem{example}{Example}

\begin{document}

\title{$L_p$-Nested Symmetric Distributions}

\author{Fabian Sinz}
\email{fabee@tuebingen.mpg.de}
\address{Max Planck Institute for Biological Cybernetics\\
       Spemannstra{\ss}e 41\\
       72076 T{\"u}bingen, Germany}
\author{Matthias Bethge}
\email{mbethge@tuebingen.mpg.de}
\address{Max Planck Institute for Biological Cybernetics\\
       Spemannstra{\ss}e 41\\
       72076 T{\"u}bingen, Germany}


\maketitle

\begin{abstract}
  Tractable generalizations of the Gaussian distribution play an
  important role for the analysis of high-dimensional data. One very
  general super-class of Normal distributions is the class of
  $\nu$-spherical distributions whose random variables can be
  represented as the product $\x = r\cdot \u$ of a uniformly
  distribution random variable $\u$ on the $1$-level set of a
  positively homogeneous function $\nu$ and arbitrary positive radial
  random variable $r$. Prominent subclasses of $\nu$-spherical
  distributions are spherically symmetric distributions
  ($\nu(\x)=\|\x\|_2$) which have been further generalized to the
  class of $L_p$-spherically symmetric distributions
  ($\nu(\x)=\|\x\|_p$). Both of these classes contain the Gaussian as
  a special case. In general, however, $\nu$-spherical distributions
  are computationally intractable since, for instance, the
  normalization constant or fast sampling algorithms are unknown for
  an arbitrary $\nu$. In this paper we introduce a new subclass of
  $\nu$-spherical distributions by choosing $\nu$ to be a nested
  cascade of $L_p$-norms. This class, which we consequently call {\em
    $L_p$-nested symmetric distributions} is still computationally
  tractable, but includes all the aforementioned subclasses as a
  special case. We derive a general expression for $L_p$-nested
  symmetric distributions as well as the uniform distribution on the
  $L_p$-nested unit sphere, including an explicit expression for the
  normalization constant. We state several general properties of
  $L_p$-nested symmetric distributions, investigate its marginals,
  maximum likelihood fitting and discuss its tight links to well known
  machine learning methods such as Independent Component Analysis
  (ICA), Independent Subspace Analysis (ISA) and mixed norm
  regularizers. Finally, we derive a fast and exact sampling algorithm
  for arbitrary $L_p$-nested symmetric distributions, and introduce
  the Nested Radial Factorization algorithm (NRF), which is a form of
  non-linear ICA that transforms any linearly mixed, non-factorial
  $L_p$-nested source into statistically independent signals.
\end{abstract}

\section{Introduction}
\label{sec:Introduction}
High-dimensional data analysis virtually always starts with the
measurement of first and second-order moments that are sufficient to
fit a multivariate Gaussian distribution, the maximum entropy
distribution under these constraints. Natural data, however, often
exhibit significant deviations from a Gaussian distribution. In order
to model these higher-order correlations, it is necessary to have more
flexible distributions available. Therefore, it is an important
challenge to find generalizations of the Gaussian distribution which
are, on the one hand, more flexible but, on the other hand, still
exhibit enough structure to be computationally and analytically
tractable.  In particular, probability models with an explicit
normalization constant are desirable because they make direct model
comparison possible by comparing the likelihood of held out test
samples for different models. Additionally, such models often allow
for a direct optimization of the likelihood.

One way of imposing structure on probability distributions is to fix
the general form of the iso-density contour lines. This approach was
taken by \cite{fernandez:1995}. They modeled the contour lines by the
level sets of a positively homogeneous function of degree one,
i.e. functions $\nu$ that fulfill $\nu(a\cdot \x)=a \cdot \nu(\x)$ for
$\x\in\R^n$ and $a\in \R_0^+$. The resulting class of $\nu$-spherical
distributions have the general form $p(\x) = \rho(\nu(\x))$ for an
appropriate $\rho$ which causes $p(\x)$ to integrate to one. Since the
only access of $\rho$ to $\x$ is via $\nu$ one can show that, for a
fixed $\nu$, those distributions are generated by a univariate radial
distribution. In other words, $\nu$-spherically distributed random
variables can be represented as a product of two independent random
variables: one positive radial variable and another variable which is
uniform on the $1$-level set of $\nu$. This property makes this class
of distributions easy to fit to data since the maximum likelihood
procedure can be carried out on the univariate radial distribution
instead of the joint density. Unfortunately, deriving the
normalization constant for the joint distribution in the general case
is intractable because it depends on the surface area of those level
sets which can usually not be computed analytically.

Known tractable subclasses of $\nu$-spherical distributions are the
Gaussian, elliptically contoured, and $L_p$-spherical
distributions. The Gaussian is a special case of elliptically
contoured distributions. After centering and whitening $\mathbf{x} :=
C^{-1/2} (\mathbf{s} - E[\mathbf{s}])$ a Gaussian distribution is
spherically symmetric and the squared $L_2$-norm $||\mathbf{x}||_2^2 =
x_1^2 + \dots + x_n^2$ of the samples follow a $\chi^2$-distribution
(i.e. the radial distribution is a $\chi$-distribution). Elliptically
contoured distributions other than the Gaussian are obtained by using
a radial distribution different from the $\chi$-distribution
\citep{kelker:1970,fang:1990}.  

The extension from $L_2$- to $L_p$-spherically symmetric distributions
is based on replacing the $L_2$-norm by the $L_p$-norm
\[\nu(\x) = \|\x\|_p = \(\sum_{i=1}^n |x_i|^p\)^\frac{1}{p}
,\:p >0\] in the definition of the density. That is, the density of
$L_p$-spherical distributions can always be written in the form $p(\x)
= \rho(||\x||_p)$.  Those distributions have been studied by
\cite{osiewalski:1993} and \cite{gupta:1997}.  We will adopt the
naming convention of \cite{gupta:1997} and call $\|\x\|_p$ an {\em
  $L_p$-norm} even though the triangle inequality only holds for $p\ge
1$.  $L_p$-spherically symmetric distribution with $p\not=2$ are no
longer invariant with respect to rotations (transformations from
$SO(n)$). Instead, they are only invariant under permutations of the
coordinate axes.  In some cases, it may not be too restrictive to
assume permutation or even rotational symmetry for the data. In other
cases, such symmetry assumptions might not be justified and let the
model miss important regularities.

Here, we present a generalization of the class of $L_p$-spherically
symmetric distribution within the class of $\nu$-spherical
distributions that makes weaker assumptions about the symmetries in
the data but still is analytically tractable. Instead of using a
single $L_p$-norm to define the contour of the density, we use a
nested cascade of $L_p$-norms where an $L_p$-norm is computed over
groups of $L_p$-norms over groups of $L_p$-norms ..., each of which
having a possibly different $p$. Due to this nested structure we call
this new class of distributions {\em $L_p$-nested symmetric
  distributions}. The nested combination of $L_p$-norms preserves
positive homogeneity but does not require permutation invariance
anymore. While $L_p$-nested distributions are still invariant under
reflections of the coordinate axes, permutation symmetry only holds
within the subspaces of the $L_p$-norms at the bottom of the
cascade. As demonstrated in \cite{sinz:2010}, one possible application
domain of $L_p$-nested symmetric distributions are patches of natural
images. In the current paper, we would like to present a formal
treatment of this class of distributions. We ask readers interested in
the application of this distributions to natural images to refer to
\cite{sinz:2010}.

We demonstrate below that the construction of the nested $L_p$-norm
cascade still bears enough structure to compute the Jacobian of
polar-like coordinates similar to those of \cite{song:1997} and
\cite{gupta:1997}. With this Jacobian at hand it is possible to
compute the univariate radial distribution for an arbitrary
$L_p$-nested density and to define the uniform distribution on the
$L_p$-nested unit sphere $\mathds L_\nu = \{\x \in \mathds R^n|
\nu(\x)=1\}$. Furthermore, we compute the surface area of the
$L_p$-nested unit sphere and, therefore, the general normalization
constant for $L_p$-nested distributions. By deriving these general
relations for the class of $L_p$-nested distributions we have
determined a new class of tractable $\nu$-spherical distributions
which is so far the only one containing the Gaussian, elliptically
contoured, and $L_p$-spherical distributions as special cases.

$L_p$-spherically symmetric distributions have been used in various
contexts in statistics and machine learning. Many results carry over
to $L_p$-nested symmetric distributions which allow a wider
application range. \cite{osiewalski:1993} showed that the posterior on
the location of a $L_p$-spherically symmetric distributions together
with an improper Jeffrey's prior on the scale does not depend on the
particular type of $L_p$-spherically symmetric distribution
used. Below, we show that this results carries over to $L_p$-nested
symmetric distributions. This means that we can robustly determine the
location parameter by Bayesian inference for a very large class of
distributions.

A large class of machine learning algorithms can be written as an
optimization problem on the sum of a regularizer and a loss
functions. For certain regularizers and loss functions, like the
sparse $L_1$ regularizer and the mean squared loss, the optimization
problem can be seen as the Maximum A Posteriori (MAP) estimate of a
stochastic model in which the prior and the likelihood are the
negative exponentiated regularizer and loss terms. Since $p(\bm x)
\propto \exp(-||\x||_p^p)$ is an $L_p$-spherically symmetric model,
regularizers which can be written in terms of a norm have a tight link
to $L_p$-spherically symmetric distributions. In an analogous way,
$L_p$-nested distributions exhibit a tight link to mixed-norm
regularizers which have recently gained increasing interest in the
machine learning community \citep[see
e.g.][]{zhao:2008,yuan:2006,kowalski:2008}. $L_p$-nested symmetric
distributions can be used for a Bayesian treatment of mixed-norm
regularized algorithms. Furthermore, they can be used to understand
the prior assumptions made by such regularizers. Below we discuss an
implicit dependence assumptions between the regularized variables that
follows from the theory of $L_p$-nested symmetric distributions.

Finally, the only factorial $L_p$-spherically symmetric distribution
\citep{sinz:2009a}, the $p$-generalized Normal distribution, has been
used as an ICA model in which the marginals follow an exponential
power distribution. This class of ICA is particularly suited for
natural signals like images and sounds
\citep{twlee:2000,zhang:2004,lewicki:2002}. Interestingly,
$L_p$-spherically symmetric distributions other than the
$p$-generalized Normal give rise to a non-linear ICA algorithm called
Radial Gaussianization for $p=2$ \citep{lyu:2009} or Radial
Factorization for arbitrary $p$ \citep{sinz:2009}. As discussed
below, $L_p$-nested distributions are a natural extension of the
linear $L_p$-spherically symmetric ICA algorithm to ISA, and give rise
to a more general non-linear ICA algorithm in the spirit of Radial
Factorization.

The remaining part of the paper is structured as follows: in Section
\ref{sec:jacobian} we define polar-like coordinates for $L_p$-nested
symmetrically distributed random variables and present an analytical
expression for the determinant of the Jacobian for this coordinate
transformation. Using this expression, we define the uniform
distribution on the $L_p$-nested unit sphere and the class of
$L_p$-nested symmetric distributions for an arbitrary $L_p$-nested
function in Section \ref{sec:distribution}. In Section
\ref{sec:marginals} we derive an analytical form of $L_p$-nested
symmetric distributions when marginalizing out lower levels of the
$L_p$-nested cascade and demonstrate that marginals of $L_p$-nested
symmetric distributions are not necessarily
$L_p$-nested. Additionally, we demonstrate that the only factorial
$L_p$-nested symmetric distribution is necessarily $L_p$-spherical and
discuss the implications of this result for mixed norm
regularizers. In Section \ref{sec:fitting} we propose an algorithm for
fitting arbitrary $L_p$-nested models and derive a sampling scheme for
arbitrary $L_p$-nested symmetric distributions. In Section
\ref{sec:robust} we generalize a result by \cite{osiewalski:1993} on
robust Bayesian inference on the location parameter to $L_p$-nested
symmetric distribution. In Section \ref{sec:isaica} we discuss the
relationship of $L_p$-nested symmetric distributions to ICA, ISA and
their possible role as prior on hidden variable in over-complete
linear models. Finally, we derive a non-linear ICA algorithm for
linearly mixed non-factorial $L_p$-nested sources in Section
\ref{sec:nonlinearICA} which we call Nested Radial Factorization
(NRF).

\section{$L_p$-nested functions, Coordinate Transformation and
  Jacobian}
\label{sec:jacobian}
Consider the function
\begin{align}
f(\x)&=\left(|x_{1}|^{p_{\emptyset}} + \left(|x_{2}|^{p_{1}}+|x_{3}|^{p_{1}}\right)^{\frac{p_{\emptyset}}{p_{1}}}\right)^{\frac{1}{p_{\emptyset}}}.
\label{eqn:exampleFct}
\end{align}
with $p_\nod,p_1\in \R^+$. This function is obviously a cascade of two
$L_p$-norms and is thus positively homogeneous of degree one. Figure
\ref{subfig:exampleTree1} shows this function visualized as a
tree. Naturally, any tree like the ones in Figure
\ref{fig:exampleTrees} corresponds to a function of the kind of
equation (\ref{eqn:exampleFct}). In general, the $n$ leaves of the
tree correspond to the $n$ coefficients of the vector $\x\in\R^n$ and
each inner node computes the $L_p$-norm of its children using its
specific $p$. We call the class of functions which is generated in
that way {\em $L_p$-nested} and the corresponding distributions, that
are symmetric or invariant with respect to it, {\em $L_p$-nested
  symmetric distributions}.

\begin{figure}[!ht]
  \begin{center}
    \subfigure[Equation (\ref{eqn:exampleFct}) as tree.]{
      \includegraphics[width=6.5cm]{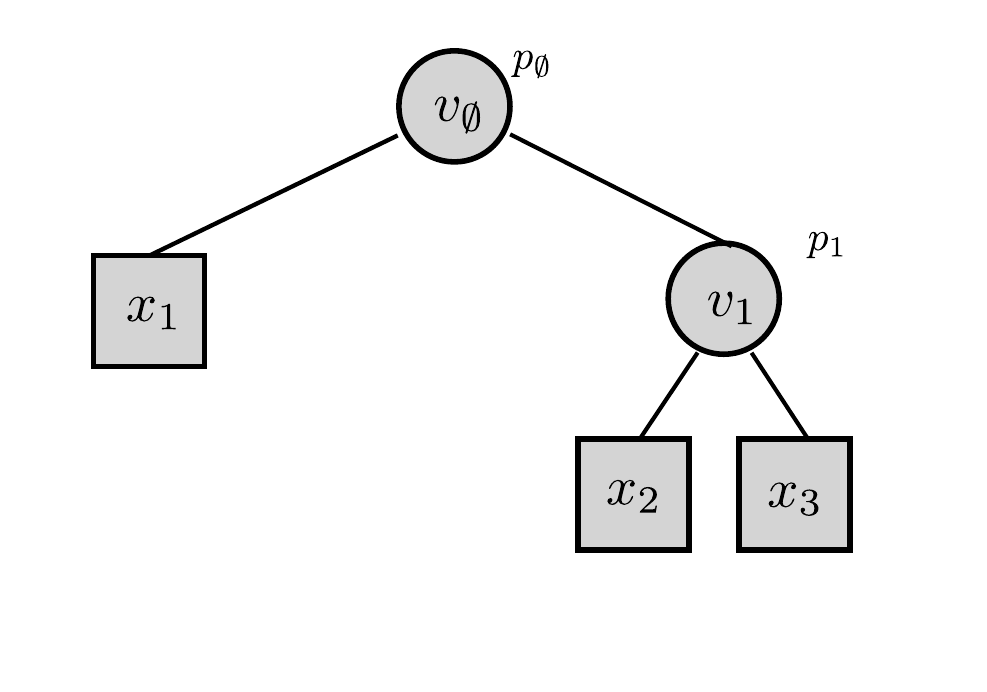}
      \label{subfig:exampleTree1}
    }
    \subfigure[Equation (\ref{eqn:exampleFct}) as tree in multi-index notation.]{
      \includegraphics[width=6.5cm]{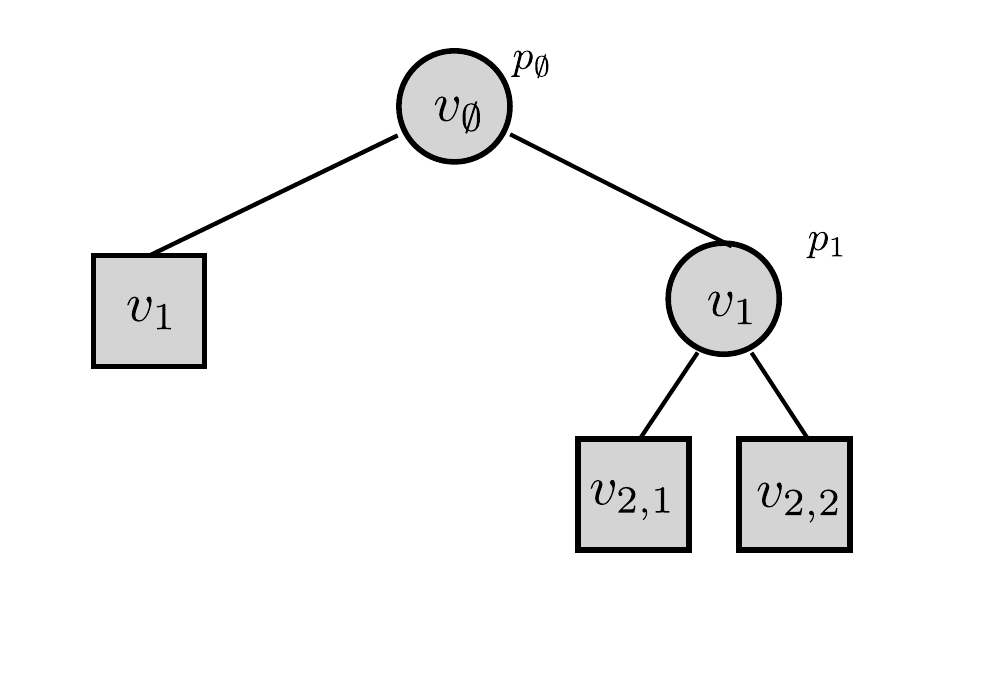}
      \label{subfig:exampleTree2}
    }
\end{center}
\caption{\label{fig:exampleTrees} Equation (\ref{eqn:exampleFct})
  visualized as a tree with two different naming conventions. Figure
  \subref{subfig:exampleTree1} shows the tree where the nodes are
  labeled with the coefficients of $\x\in\R^n$. Figure
  \subref{subfig:exampleTree2} shows the same tree in multi-index
  notation where the multi-index of a node describes the path from the
  root node to that node in the tree. The leaves $\f_1, \f_{2,1}$ and
  $\f_{2,2}$ still correspond to $x_1, x_2$ and $x_3$, respectively,
  but have been renamed to the multi-index notation used in this
  article.}
\end{figure}

$L_p$-nested functions are much more flexible in creating different
shapes of level sets than single $L_p$-norms. Those level sets become
the iso-density contours in the family of $L_p$-nested symmetric
distributions. Figure \ref{fig:contour_zoo} shows a variety of
contours generated by the simplest non-trivial $L_p$-nested function
shown in equation (\ref{eqn:exampleFct}). The shapes show the unit
spheres for all possible combinations of
$p_\nod,p_1\in\{0.5,1,2,10\}$. On the diagonal, $p_\nod$ and $p_1$ are
equal and therefore constitute $L_p$-norms. The corresponding
distributions are members of the $L_p$-spherically symmetric class.

\begin{figure}[!h]
  \begin{center}
    \includegraphics[width=15.5cm]{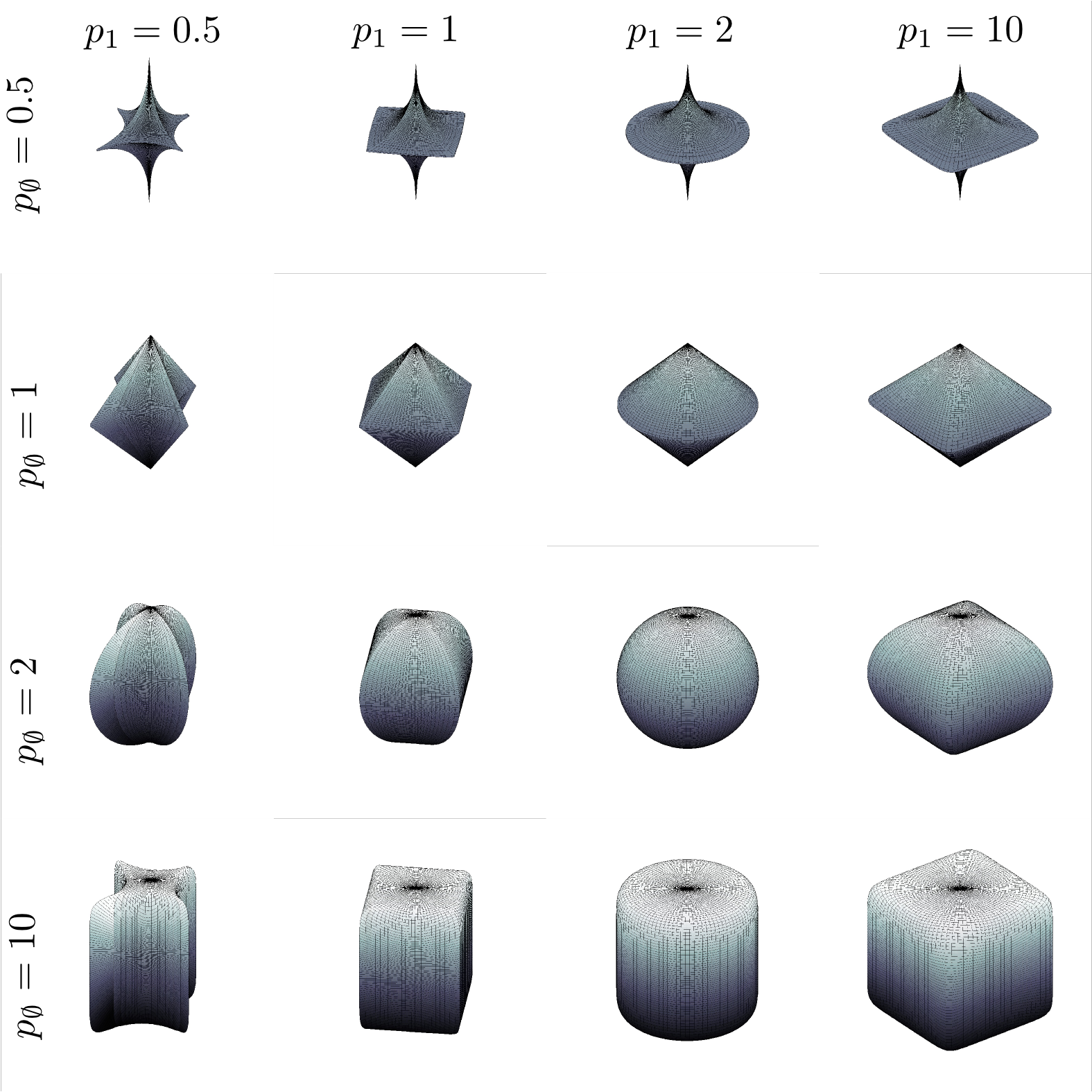}
  \end{center}
  \caption{\label{fig:contour_zoo} Variety of contours created by the
     $L_p$-nested function of equation (\ref{eqn:exampleFct}) for all combinations of
    $p_\nod,p_1\in \{0.5,1,2,10\}$.}
\end{figure}

In order to make general statements about general $L_p$-nested
functions, we introduce a notation that is suitable for the tree
structure of $L_p$-nested functions. As we will heavily use that
notation in the remainder of the paper, we would like to emphasize the
importance of the following paragraphs. We will illustrate the
notation with an example below. Additionally, Figure
\ref{fig:exampleTrees} and Table \ref{tab:notation} can be used for
reference.
\begin{table}
\begin{centering}
\renewcommand{\arraystretch}{1.2}
\begin{tabular}{ll}
\hline 
$f(\cdot)=f_{\emptyset}(\cdot)$ & $L_{p}$-nested function\tabularnewline
\hline 
$I=i_{1},...,i_{m}$ & Multi-index denoting a node in the tree. The single indices describe \tabularnewline
 & the path from the root node to the respective node $I$.\tabularnewline
\hline 
$\bm{x}_{I}$ & All entries in $\bm{x}$ that correspond to the leaves in the subtree
under\tabularnewline
 & the node $I$.\tabularnewline
\hline 
$\bm{x}_{\widehat I}$ & All entries in $\bm{x}$ that are not leaves in the subtree
under\tabularnewline
 & the node $I$.\tabularnewline
\hline 
$f_{I}(\cdot)$ & $L_{p}$-nested function corresponding to the subtree under the node
$I$.\tabularnewline
\hline 
$\f_{\emptyset}$  & Function value at the root node\tabularnewline
\hline 
$\f_{I}$ & Function value at an arbitrary node with multi-index $I$.\tabularnewline
\hline 
$\ell_{I}$ & The number of direct children of a node $I$.\tabularnewline
\hline 
$n_{I}$ & The number of leaves in the subtree under the node $I$.\tabularnewline
\hline 
$\ff_{I,1:\ell_{I}}$ & Vector with the function values at the direct children of a node $I$.\tabularnewline
\hline
\end{tabular}
\par\end{centering}

\caption{\label{tab:notation} Summary of the notation used for $L_{p}$-nested functions in this
article.}

\end{table}

We use multi-indices to denote the different nodes of the tree
corresponding to an $L_p$-nested function $f$. The function $f=f_\nod$
itself computes the value $\f_\nod$ at the root node (see Figure
\ref{fig:exampleTrees}). Those values are denoted by variables $\f$.
The functions corresponding to its children are denoted by
$f_1,...,f_{\ell_\nod}$, i.e.  $f(\cdot) = f_\nod(\cdot) =
\|(f_1(\cdot),...,f_{\ell_\nod}(\cdot))\|_{p_\nod}$. We always use the
letter ``$\ell$'' indexed by the node's multi-index to denote the
total number of direct children of that node. The functions of the
children of the $i$\th child of the root node are denoted by
$f_{i,1},...,f_{i,\ell_i}$ and so on.  In this manner, an index is
added for denoting the children of a particular node in the tree and
each multi-index denotes the path to the respective node in the tree.
For the sake of compact notation, we use upper case letters to denote
a single multi-index $I = i_1,...,i_\ell$. The range of the single
indices and the length of the multi-index should be clear from the
context.  A concatenation $I,k$ of a multi-index $I$ with a single
index $k$ corresponds to adding $k$ to the index tuple, i.e.
$I,k=i_1,...,i_m,k$.  We use the convention that $I,\nod=I$.  Those
coefficients of the vector $\x$ that correspond to leaves of the
subtree under a node with the index $I$ are denoted by $\x_I$. The
complement of those coefficients, i.e. the ones that are not in the
subtree under the node $I$, are denoted by $\x_{\widehat I}$. The
number of leaves in a subtree under a node $I$ is denoted by $n_I$. If
$I$ denotes a leaf then $n_I=1$.

The $L_p$-nested function associated with the subtree under a node $I$
is denoted by $$f_I(\x_I) =
||(f_{I,1}(\x_{I,1}),...,f_{I,\ell_I}(\x_{I,\ell_I}))^\top ||_{p_I}.$$
Just like for the root node, we use the variable $\f_I$ to denote the
function value $\f_I=f_I(\x_I)$ of a subtree $I$. A vector with the
function values of the children of $I$ is denoted with bold font
$\ff_{I,1:\ell_I}$ where the colon indicates that we mean the vector
of the function values of the $\ell_I$ children of node $I$:
\begin{align*}
  f_I(\x_I) &= ||(f_{I,1}(\x_{I,1}),...,f_{I,\ell_I}(\x_{I,\ell_I}))^\top ||_{p_I}\\
  &= ||(\f_{I,1},...,\f_{I,\ell_I} )^\top||_{p_I} = ||\ff_{I,1:\ell_I}||_{p_I}.
\end{align*}

Note that we can assign an arbitrary $p$ to leaf nodes since $p$ for
single variables always cancel. For that reason we can choose an
arbitrary $p$ for convenience and fix its value to $p=1$. Figure
\ref{subfig:exampleTree2} shows the multi-index notation for our
example of equation (\ref{eqn:exampleFct}).

To illustrate the notation: Let $I=i_1,...,i_d$ be the multi-index of
a node in the tree. $i_1,...,i_d$ describes the path to that node,
i.e. the respective node is the $i_d^{th}$ child of the $i_{d-1}^{th}$
child of the $i_{d-2}^{th}$ child of the ... of the $i_1^{th}$ child
of the root node. Assume that the leaves in the subtree below the node
$I$ cover the vector entries $x_2,...,x_{10}$. Then
$\x_I=(x_2,...,x_{10})$, $\x_{\widehat I}=(x_1,x_{11},x_{12},...)$, and $n_{I}=9$. Assume that node $I$ has
$\ell_I=2$ children. Those would be denoted by $I,1$ and $I,2$. The
function realized by node $I$ would be denoted by $f_I$ and only acts
on $\x_I$. The value of the function would be $f_I(\x_I)=\f_I$ and the
vector containing the values of the children of $I$ would be
$\ff_{I,1:2}=(\f_{I,1},\f_{I,2})^\top=(f_{I,1}(\x_{I,1}),f_{I,2}(\x_{I,2}))^\top$.

We now introduce a coordinate representation that is especially
tailored to $L_p$-nested symmetrically distributed variables:  One of
the most important consequence of the positive homogeneity of $f$ is
that it can be used to ``normalize'' vectors and, by that property,
create a polar like coordinate representation of a vector $\x$. Such
polar-like coordinates generalize the coordinate representation for
$L_p$-norms by \cite{gupta:1997}.

\begin{definition}[Polar-like Coordinates]
  \label{def:coordinates}
  We define the following polar-like coordinates
  for a vector $\x\in\R^n$:
  \begin{align*}
    u_{i} &= \frac{x_{i}}{f(\x)}\mbox{ for }i=1,...,n-1\\
    r &= f(\x).
  \end{align*}
  The inverse coordinate transformation is given by 
  \begin{align*}
    x_{i} &= ru_{i}\mbox{ for }i=1,...,n-1\\
    x_{n} &= r\Delta_n u_{n}
  \end{align*}
  where $\Delta_n=\mathrm{sgn}\, x_{n}$ and $u_{n}=\frac{|x_n|}{f(\x)}$.
\end{definition}

Note that $u_n$ is not part of the coordinate representation since
normalization with $1/f(\x) $ decreases the degrees of
freedom $\u$ by one, i.e. $u_n$ can always be computed from all other
$u_i$ by solving $f(\u)=f\(\x/f(\x)\)=1$ for $u_n$. We only use the
term $u_n$ for notational simplicity. With a slight abuse of notation,
we will use $\u$ to denote the normalized vector $\x/f(\x)$ or only
its first $n-1$ components. The exact meaning should always be clear
from the context.

The definition of the coordinates is exactly the same as the one by
\cite{gupta:1997} with the only difference that the $L_p$-norm is
replaced by an $L_p$-nested function. Just as in the case of
$L_p$-spherical coordinates, it will turn out that the determinant of
the Jacobian of the coordinate transformation does not depend on the
value of $\Delta_n$ and can be computed analytically. The determinant
is essential for deriving the uniform distribution on the unit
$L_p$-nested sphere $\mathds L_f$, i.e. the $1$-level set of
$f$. Apart from that, it can be used to compute the radial
distribution for a given $L_p$-nested distribution. We start by
stating the general form of the determinant in terms of the partial
derivatives $\frac{\partial u_n}{\partial u_k}$, $u_k$ and
$r$. Afterwards we demonstrate that those partial derivatives have a
special form and that most of them cancel in Laplace's expansion of
the determinant.

\begin{lemma}[Determinant of the Jacobian]
  \label{lem:generalDeterminant}
  Let $r$ and $\u$ be defined as in Definition
  \ref{def:coordinates}. The general form of the determinant of the
  Jacobian $\mathcal J=\(\frac{\partial x_i}{\partial y_j}\)_{ij}$ of
  the inverse coordinate transformation for $y_1=r$ and $y_i=u_{i-1}$
  for $i=2,...,n$, is given by
  \begin{align}
    |\det\mathcal{J}| &= r^{n-1}\left(-\sum_{k=1}^{n-1}\frac{\partial
        u_{n}}{\partial u_{k}}\cdot
      u_{k}+u_{n}\right)\label{eq:generalDeterminant}.
  \end{align}
\end{lemma}
\begin{proof}The proof can be found in the Appendix \ref{app:determinantOfJacobian}.
\end{proof}

The problematic part in equation (\ref{eq:generalDeterminant}) are the
terms $\frac{\partial u_{n}}{\partial u_{k}}$, which obviously involve
extensive usage of the chain rule. Fortunately, most of them cancel when
inserting them back into equation (\ref{eq:generalDeterminant}),
leaving a comparably simple formula. The remaining part of this
section is devoted to computing those terms and demonstrating how they
vanish in the formula for the determinant. Before we state the general
case we would like to demonstrate the basic mechanism through a simple
example. We urge the reader to follow the next example as it
illustrates all important ideas about the coordinate transformation
and its Jacobian.

\begin{example}
\label{ex:determinantEx1}
Consider an $L_p$-nested function very similar to our introductory
example of equation (\ref{eqn:exampleFct}):
 $$f(\x)=\left(\left(|x_{1}|^{p_{1}}+|x_{2}|^{p_{1}}\right)^{\frac{p_{\emptyset}}{p_{1}}}+|x_{3}|^{p_{\emptyset}}\right)^{\frac{1}{p_{\emptyset}}}.$$ 
 Setting $\u = \frac{\x}{f(\x)}$ and solving for $u_3$ yields
\begin{align}
  f(\u) = 1 &\Leftrightarrow\:\:\: u_3 =
  \left(1-\left(|u_{1}|^{p_{1}}+|u_{2}|^{p_{1}}\right)^{\frac{p_{\emptyset}}{p_{1}}}\right)^{\frac{1}{p_{\emptyset}}} \label{eq:u3}
\end{align}
We would like to emphasize again, that $u_3$ is actually not part of
the coordinate representation and only used for notational
simplicity. By construction, $u_3$ is always positive. This is no
restriction since Lemma \ref{sec:jacobian} shows that the determinant
of the Jacobian does not depend on its sign. However, when computing
the volume and the surface area of the $L_p$-nested unit sphere it
will become important since it introduces a factor of $2$ to account
for the fact that $u_3$ (or $u_n$ in general) can in principle also
attain negative values. 

Now, consider
 \begin{align*}
 \G_2(\u_{\widehat 2}) &= g_{2}(\u_{\widehat 2})^{1-p_{\emptyset}} =  \left(1-\left(|u_{1}|^{p_{1}}+|u_{2}|^{p_{1}}\right)^{\frac{p_{\emptyset}}{p_{1}}}\right)^{\frac{1-p_{\emptyset}}{p_{\emptyset}}}\\
 \F_1(\u_{1}) &= f_{1}(\u_{1})^{p_{\emptyset}-p_{1}} =  \left(|u_{1}|^{p_{1}}+|u_{2}|^{p_{1}}\right)^{\frac{p_{\emptyset}-p_{1}}{p_{1}}},
\end{align*}
where the subindices of $\u,\:f,\:g,\:\G$ and $\F$ have to be read as
multi-indices. The function $g_I$ computes the value of the node $I$
from all other leaves that are not part of the subtree under $I$
by fixing the value of the root node to one. 

$\G_{2}(\bm u_{\widehat 2})$ and $\F_{1}(\u_{1})$ are terms that arise
from applying the chain rule when computing the partial derivatives
$\frac{\partial u_3}{\partial u_k}$. Taking those partial
derivatives can be thought of as pealing off layer by layer of Equation
\eqref{eq:u3} via the chain rule. By doing so, we ``move'' on a path
between $u_3$ and $u_k$. Each application of the chain rule
corresponds to one step up or down in the tree. First, we move upwards
in the tree, starting from $u_3$. This produces the $\G$-terms. In
this example, there is only one step upwards, but in general, there
can be several, depending on the depth of $u_n$ in the tree. Each step
up will produce one $\G$-term. At some point, we will move downwards
in the tree to reach $u_k$. This will produce the $\F$-terms. While
there are as many $\G$-terms as upward steps, there is one term less
when moving downwards. Therefore, in this example, there is one term
$\G_{2}(\bm u_{\widehat 2})$ which originates from using the chain
rule upwards in the tree and one term $\F_{1}(\u_{1})$ from using it
downwards. The indices correspond to the multi-indices of the
respective nodes. 

Computing the derivative yields
\begin{align*}
\frac{\partial u_3}{\partial u_{k}} &= -\G_{2}(\bm u_{\widehat 2})\F_{1}(\u_{1})\Delta_{k}|u_{k}|^{p_{1}-1}.
\end{align*}
By inserting the results in equation (\ref{eq:generalDeterminant}) we
obtain
\begin{align*}
  \frac{1}{r^2}|\mathcal J|&=  \sum_{k=1}^{2}\G_{2}(\bm u_{\widehat 2})\F_{1}(\u_{1})|u_{k}|^{p_{1}}+u_3\\
    &= \G_{2}(\bm u_{\widehat 2}) \( \F_{1}(\u_{1})\sum_{k=1}^{2}|u_{k}|^{p_{1}}
    +1-\F_1(\u_{1})\F_{1}(\u_{1})^{-1}\left(|u_{1}|^{p_{1}}+|u_{2}|^{p_{1}}\right)^{\frac{p_{\emptyset}}{p_{1}}} \)\\
    &= \G_{2}(\bm u_{\widehat 2}) \( \F_{1}(\u_{1})\sum_{k=1}^{2}|u_{k}|^{p_{1}}
    +1-\F_1(\u_{1})\sum_{k=1}^{2}|u_{k}|^{p_{1}}  \)\\
    &= \G_2(\bm u_{\widehat 2}).
\end{align*}
\end{example}

The example suggests that the terms from using the chain rule
downwards in the tree cancel while the terms from using the chain rule
upwards remain.  The following proposition states that this is true in
general.

\begin{proposition}[Determinant of the Jacobian]
  \label{pro:DetJacobian} Let $\mathcal L$ be the set of multi-indices
  of the path from the leaf $u_n$ to the root node (excluding the root
  node) and let the terms $\G_{I,\ell_I}(\bm u_{\widehat{I,\ell_I}})$ recursively be defined as
\begin{align}
  \G_{I,\ell_I}(\bm u_{\widehat{I,\ell_I}}) &=  \g_{I,{\ell_I}}(\u_{\widehat{I,\ell_I}})^{p_{I,{\ell_I}}-p_{I}}  =\left(\g_{I}(\u_{\widehat{I}})^{p_{I}}-\sum_{j=1}^{\ell-1}f_{I,j}(\u_{I,j})^{p_{I}}\right)^{\frac{p_{I,{\ell_I}}-p_{I}}{p_{I}}} \label{eq:defF},
\end{align} 
where each of the functions $g_{I,\ell_I}$ computes the value of the
$\ell$\th child of a node $I$ as a function of its neighbors $(I,1)$,
$...$, $(I,\ell_I-1)$ and its parent $I$ while fixing the value of the
root node to one. This is equivalent to computing the value of the
node $I$ from all coefficients $\u_{\widehat{I}}$ that are not leaves
in the subtree under $I$. Then, the determinant of the Jacobian for an
$L_p$-nested function is given by
\begin{align*}
\det|\mathcal{J}| & = r^{n-1}\prod_{L\in\mathcal L}\G_L(\bm u_{\widehat L}).
\end{align*}
\end{proposition}
\begin{proof}
The proof can be found in the Appendix \ref{app:determinantOfJacobian}.
\end{proof}
Let us illustrate the determinant with two examples:  
\begin{example}
 Consider a normal
  $L_p$-norm $$f(\x)=\(\sum_{i=1}^n|x_i|^p\)^\frac{1}{p}$$ which is
  obviously also an $L_p$-nested function.  Resolving the equation for
  the last coordinate of the normalized vector $\u$ yields
  $g_n(\u_{\widehat{n}})=u_n=\(1-\sum_{i=1}^{n-1}|u_i|^p\)^\frac{1}{p}$. Thus, the term
  $\G_n(\bm u_{\widehat n})$ term is given by
  $\(1-\sum_{i=1}^{n-1}|u_i|^p\)^\frac{1-p}{p}$ which yields a
  determinant of $|\det \mathcal
  J|=r^{n-1}\(1-\sum_{i=1}^{n-1}|u_i|^p\)^\frac{1-p}{p}$. This is
  exactly the one derived by \cite{gupta:1997}.
\end{example}

\begin{example}
  Consider the introductory example
\begin{align*}
f(\x)&=\left(|x_{1}|^{p_{\emptyset}} + \left(|x_{2}|^{p_{1}}+|x_{3}|^{p_{1}}\right)^{\frac{p_{\emptyset}}{p_{1}}}\right)^{\frac{1}{p_{\emptyset}}}.
\end{align*}
Normalizing and resolving  for the last coordinate yields 
\begin{align*}
u_3 &= \( \(1- |u_1|^{p_\nod} \)^\frac{p_1}{p_\nod} -
|u_2|^{p_1}  \)^\frac{1}{p_1}
\end{align*}
and the terms $\G_2(\bm u_{\widehat 2})$ and $\G_{2,2}(\bm u_{\widehat{2,2}})$ of the determinant $|\det \mathcal
J| = r^{2} \G_{2}(\bm u_{\widehat 2})\G_{2,2}(\bm u_{\widehat{2,2}})$ are given by
\begin{align*}
\G_{2}(\bm u_{\widehat 2}) &=   \(1- |u_1|^{p_\nod} \)^\frac{p_1-p_\nod}{p_\nod}\\
\G_{2,2}(\bm u_{\widehat{2,2}}) &= \( \(1- |u_1|^{p_\nod} \)^\frac{p_1}{p_\nod} -
|u_2|^{p_1}  \)^\frac{1-p_1}{p_1}.
\end{align*}
Note the difference to Example \ref{ex:determinantEx1} where $x_3$ was
at depth one in the tree while $x_3$ is at depth two in the current
case. For that reason, the determinant of the Jacobian in Example
\ref{ex:determinantEx1} only involved one $\G$-term while it has two
$\G$-terms here.
\end{example}

\section{$L_p$-Nested Symmetric and $L_p$-Nested Uniform Distribution}
\label{sec:distribution}
In this section, we define the $L_p$-nested symmetric and the
$L_p$-nested uniform distribution and derive their partition
functions. In particular, we derive the surface area of an arbitrary
$L_p$-nested unit sphere $\mathds L_f = \{\x \in \R^n\;|\; f(\x)=1\}$ corresponding to an
$L_p$-nested function $f$. By equation (5) of \cite{fernandez:1995}
every $\nu$-spherically symmetric and hence any $L_p$-nested density
has the form
\begin{align}
\rho (\x) = \frac{\varrho(f(\x))}{f(\x)^{n-1}\mathcal S_f(1)},\label{eqn:generalNuSphericalDistribution}
\end{align}
where $\mathcal S_f$ is the surface area of $\mathds L_f$ and
$\varrho$ is a density on $\R^+$. Thus, we need to compute the surface
area of an arbitrary $L_p$-nested unit sphere to obtain the partition
function of equation (\ref{eqn:generalNuSphericalDistribution}).

\begin{proposition}[Volume and Surface of the $L_p$-nested Sphere]
  \label{pro:VolumeSurface}
  Let $f$ be an $L_p$-nested function and let $\mathcal I$ be the set
  of all multi-indices denoting the inner nodes of the tree structure
  associated with $f$.  The volume $\mathcal V_f(R)$ and the surface
  $\mathcal S_f(R)$ of the $L_p$-nested sphere with radius $R$ are
  given by
\begin{align}
\mathcal V_f(R) &= \frac{R^n 2^n}{n}\prod_{I\in \mathcal
  I}\frac{1}{p_I^{\ell_I-1}}\prod_{k=1}^{\ell_I-1}
B\left[\frac{\sum_{i=1}^{k}
    n_{I,k}}{p_I},\frac{n_{I,k+1}}{p_I}\right] \label{eq:volbeta}\\
&= \frac{R^n 2^n}{n}\prod_{I\in \mathcal
  I}\frac{\prod_{k=1}^{\ell_I}\Gamma\left[\frac{n_{I,k}}{p_I}\right]}{p_I^{\ell_I-1}\Gamma\left[\frac{n_I}{p_I}\right]}\label{eq:volgamma}
\\
\mathcal S_f(R) &= R^{n-1} 2^n\prod_{I\in \mathcal
  I}\frac{1}{p_I^{\ell_I-1}}\prod_{k=1}^{\ell_I-1}
B\left[\frac{\sum_{i=1}^{k}
    n_{I,k}}{p_I},\frac{n_{I,k+1}}{p_I}\right]\label{eq:surbeta}\\
&= R^{n-1} 2^n\prod_{I\in \mathcal
  I}\frac{\prod_{k=1}^{\ell_I}\Gamma\left[\frac{n_{I,k}}{p_I}\right]}{p_I^{\ell_I-1}\Gamma\left[\frac{n_I}{p_I}\right]}\label{eq:surgamma}
\end{align}
where $B[a,b]=\frac{\Gamma[a]\Gamma[b]}{\Gamma[a+b]}$ denotes the
$\beta$-function. 
\end{proposition}
\begin{proof}
The proof can be found in the Appendix \ref{app:VolumeSurface}.
\end{proof}

Inserting the surface area in equation
\ref{eqn:generalNuSphericalDistribution}, we obtain the general form
of an $L_p$-nested symmetric distribution for any given radial density $\varrho$.

\begin{corollary}[$L_p$-nested Symmetric Distribution]
\label{cor:lpnestedSymDistr}
Let $f$ be an $L_p$-nested function and $\varrho$ a density on
$\R^+$. The corresponding $L_p$-nested symmetric distribution is given
by
\begin{align}
 \rho(\x) &= \frac{\varrho(f(\x))}{f(\x)^{n-1}\mathcal S_f(1)}\nonumber\\
 &=  \frac{\varrho(f(\x))}{2^n f(\x)^{n-1}}\prod_{I\in \mathcal
  I}p_I^{\ell_I-1}\prod_{k=1}^{\ell_I-1}
B\left[\frac{\sum_{i=1}^{k}
    n_{I,k}}{p_I},\frac{n_{I,k+1}}{p_I}\right]^{-1} \label{eq:generalLpNested}.
\end{align}
\end{corollary}

The results of \cite{fernandez:1995} imply that for any
$\nu$-spherically symmetric distribution, the radial part is
independent of the directional part, i.e. $r$ is independent of
$\u$. The distribution of $\u$ is entirely determined by the choice of
$\nu$, or by the $L_p$-nested function $f$ in our case. The
distribution of $r$ is determined by the radial density
$\varrho$. Together, an $L_p$-nested symmetric distribution is
determined by both, the $L_p$-nested function $f$ and the choice
of $\varrho$. 
From equation (\ref{eq:generalLpNested}), we can see that its density
function must be the inverse of the surface area of $\mathds L_f$ times
the radial density when transforming
(\ref{eqn:generalNuSphericalDistribution}) into the coordinates of
Definition \ref{def:coordinates} and separating $r$ and $\u$ (the
factor $f(\x)^{n-1}=r$ cancels due to the determinant of the
Jacobian). For that reason we call the distribution of $\u$ {\em
  uniform on the $L_p$-sphere $\mathds L_f$} in analogy to
\cite{song:1997}. Next, we state its form in terms of the coordinates
$\u$.

\begin{proposition}[$L_p$-nested Uniform Distribution]
\label{pro:LpNestedUniform}
Let $f$ be an $L_p$-nested function. Let $\mathcal L$ be set set of
multi-indices on the path from the root node to the leaf corresponding
to $x_n$. The uniform distribution on the $L_p$-nested unit sphere,
i.e. the set $\mathds L_f = \{\x\in\R^n|f(\x)=1\}$ is given by the
following density over $u_1,...,u_{n-1}$
\begin{align*}
\rho(u_1,,...,u_{n-1}) &=  \frac{\prod_{L\in\mathcal L}\G_L(\bm u_{\widehat L})}{2^{n-1}}\prod_{I\in \mathcal
  I} p_I^{\ell_I-1} \prod_{k=1}^{\ell_I-1}
B\left[\frac{\sum_{i=1}^{k}
    n_{I,k}}{p_I},\frac{n_{I,k+1}}{p_I}\right]^{-1} 
\end{align*}
\end{proposition}
\begin{proof}
  Since the $L_p$-nested sphere is a measurable and compact set, the
  density of the uniform distribution is simply one over the surface
  area of the unit $L_p$-nested sphere. The surface $\mathcal S_f(1)$
  is given by Proposition \ref{pro:VolumeSurface}. Transforming
  $\frac{1}{\mathcal S_f(1)}$ into the coordinates of Definition
  \ref{def:coordinates} introduces the determinant of the Jacobian
  from Proposition \ref{pro:DetJacobian} and an additional factor of
  $2$ since the $(u_1,...,u_{n-1})\in\mathds R^{n-1}$ have to account
  for both half-shells of the $L_p$-nested unit sphere, i.e. to
  account for the fact that $u_n$ could have been be positive or
  negative. This yields the expression above.
\end{proof}

\begin{example}
Let us again demonstrate the proposition at the special case where $f$
is an $L_p$-norm $f(\x) = ||\x||_p = \(\sum_{i=1}^n
|x_i|^p\)^\frac{1}{p}$.  Using Proposition \ref{pro:VolumeSurface},
the surface area is given by
\begin{align*}
\mathcal S_{||\cdot||_p} &= 2^n\frac{1}{p_\nod^{\ell_\nod-1}}\prod_{k=1}^{\ell_\nod-1}
B\left[\frac{\sum_{i=1}^{k}
    n_{k}}{p_\nod},\frac{n_{k+1}}{p_\nod}\right]
= \frac{ 2^n\Gamma^n\left[\frac{1}{p}\right]}{p^{n-1}\Gamma\left[\frac{n}{p}\right]}.
\end{align*}
The factor $\G_n(\bm u_{\widehat n})$ is given by $\(1-\sum_{i=1}^{n-1}
|u_i|^p\)^\frac{1-p}{p}$ (see the $L_p$-norm example before), which,
after including the factor $2$, yields the uniform distribution on the
$L_p$-sphere as defined in \cite{song:1997}
$$p(\u)= \frac{p^{n-1}\Gamma\left[\frac{n}{p}\right]}{2^{n-1}\Gamma^n\left[\frac{1}{p}\right]} \(1-\sum_{i=1}^{n-1}
|u_i|^p\)^\frac{1-p}{p}.$$
\end{example}

\begin{example}
  \label{ex:uniformDistributionUnitBall}
  As a second illustrative example, we consider the uniform density on
  the $L_p$-nested unit ball, i.e. the set $\{\x\in\R^n|\:f(\x)\le
  1\}$, and derive its radial distribution $\varrho$. The density of
  the uniform distribution on the unit $L_p$-nested ball does not
  depend on $\x$ and is given by $\rho(\x)=1/\mathcal
  V_f(1)$. Transforming the density into the polar-like coordinates
  with the determinant from Proposition \ref{pro:DetJacobian} yields
\begin{align*}
\frac{1}{\mathcal V_f(1)} &= \frac{nr^{n-1} \prod_{L\in\mathcal L}\G_L(\bm u_{\widehat L})}{2^{n-1}}\prod_{I\in \mathcal
  I} p_I^{\ell_I-1} \prod_{k=1}^{\ell_I-1}
B\left[\frac{\sum_{i=1}^{k}
    n_{I,k}}{p_I},\frac{n_{I,k+1}}{p_I}\right]^{-1}.
\end{align*}
After separating out the uniform distribution on the $L_p$-nested unit
sphere, we obtain the radial distribution
\begin{align*}
\varrho(r) &= nr^{n-1} \mbox{ for  } 0<r\le 1
\end{align*} which is a $\beta$-distribution with parameters $n$ and
$1$. 
\end{example}

The radial distribution from the preceding example is of great
importance for our sampling scheme derived in Section
\ref{sec:sampling}. The idea behind it is the following: First, a
sample from an ``simple'' $L_p$-nested distribution is drawn. Since
the radial and the uniform component on the $L_p$-nested unit sphere
are statistically independent, we can get a sample from the uniform
distribution on the $L_p$-nested unit sphere by simply normalizing the
sample from the simple distribution. Afterwards we can multiply it
with a radius drawn from the radial distribution of the $L_p$-nested
distribution that we actually want to sample from. The role of the
simple distribution will be played by the uniform distribution within
the $L_p$-nested unit ball. Sampling from it is basically done by
applying the steps in proof of Proposition \ref{pro:VolumeSurface}
backwards. We lay out the sampling scheme in more detail in Section
\ref{sec:sampling}.

\section{Marginals}
\label{sec:marginals}

In this section we discuss two types of marginals: First, we
demonstrate that, in contrast to $L_p$-spherically symmetric
distributions, marginals of $L_p$-nested distributions are not
necessarily $L_p$-nested again. The second type of marginals we
discuss are obtained by collapsing all leaves of a subtree into the
value of the subtree's root node. For that case we derive an
analytical expression and show that the values of the root node's
children follow a special kind of Dirichlet distribution.

\cite{gupta:1997} show that marginals of $L_p$-spherically symmetric
distributions are again $L_p$-spherically symmetric. This does not
hold, however, for $L_p$-nested symmetric distributions. This can be
shown by a simple counterexample. Consider the $L_p$-nested function
 $$f(\x)=\left(\left(|x_{1}|^{p_{1}}+|x_{2}|^{p_{1}}\right)^{\frac{p_{\emptyset}}{p_{1}}}+|x_{3}|^{p_{\emptyset}}\right)^{\frac{1}{p_{\emptyset}}}.$$ 
 The uniform distribution inside the $L_p$-nested ball corresponding
 to $f$ is given by
$$\rho(\x) =
\frac{np_1p_\nod\Gamma\left[\frac{2}{p_1}\right]\Gamma\left[\frac{3}{p_\nod}\right]}{2^3\Gamma^2\left[\frac{1}{p_1}\right]
\Gamma\left[\frac{2}{p_0}\right]\Gamma\left[\frac{1}{p_0}\right]}.$$
The marginal $\rho(x_1,x_3)$ is given by
\begin{align*}
\rho(x_1,x_3) &= \frac{np_1p_\nod\Gamma\left[\frac{2}{p_1}\right]\Gamma\left[\frac{3}{p_\nod}\right]}{2^3\Gamma^2\left[\frac{1}{p_1}\right]
\Gamma\left[\frac{2}{p_0}\right]\Gamma\left[\frac{1}{p_0}\right]} \(\(1-|x_3|^{p_\nod}\)^{\frac{p_1}{p_\nod}}-|x_1|^{p_1}\)^\frac{1}{p_1}.
\end{align*}
This marginal is $L_p$-spherically symmetric. Since any $L_p$-nested
distribution in two dimensions must be $L_p$-spherically symmetric it
cannot be $L_p$-nested symmetric as well. Figure
\ref{fig:marginalExample} shows a scatter plot of the marginal
distribution. Besides the fact that the marginals are not contained in
the family of $L_p$-nested distributions, it is also hard to derive a
general form for them. This is not surprising given that the general
form of marginals for $L_p$-spherically symmetric distributions
involves an integral that cannot be solved analytically in general and
is therefore not very useful in practice \citep{gupta:1997}. For that
reason we cannot expect marginals of $L_p$-nested symmetric
distributions to have a simple form.

\begin{figure}[!ht]
  \begin{center}
    \includegraphics[width=15cm]{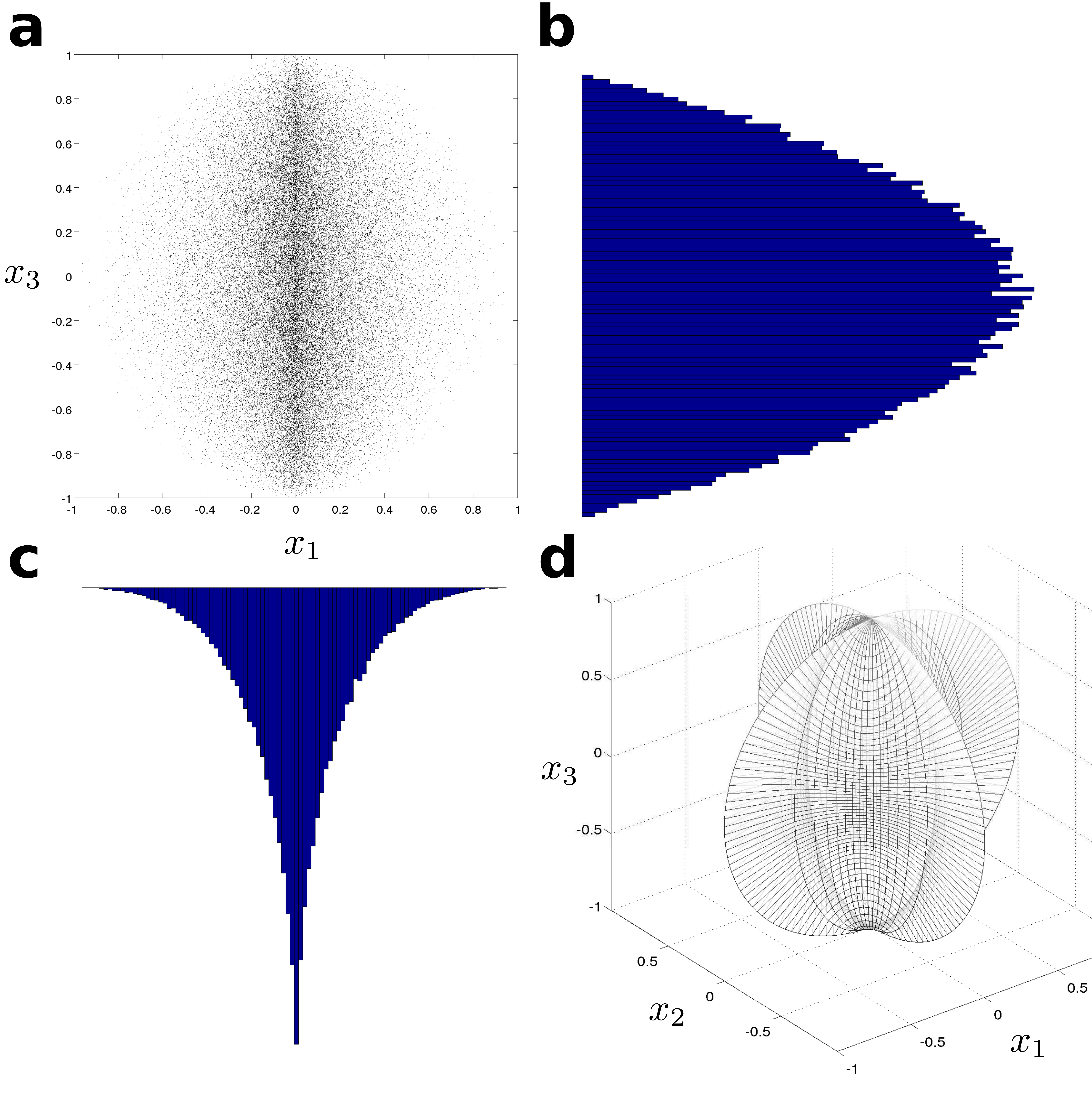}
  \end{center}
  \caption{\label{fig:marginalExample} Marginals of $L_p$-nested
    symmetric distributions are not necessarily $L_p$-nested
    symmetric: Figure ({\bf a}) shows a scatter plot of the
    $(x_1,x_2)$-marginal of the counterexample in the text with
    $p_\nod=2$ and $p_1=\frac{1}{2}$.  Figure ({\bf d}) displays the
    corresponding $L_p$-nested sphere.  ({\bf b-c}) show the
    univariate marginals for the scatter plot. Since any
    two-dimensional $L_p$-nested distribution must be $L_p$-spherical,
    the marginals should be identical. This is clearly not the
    case. Thus, ({\bf a}) is not $L_p$-nested symmetric.}
\end{figure}

In contrast to single marginals, it is possible to specify the joint
distribution of leaves and inner nodes of an $L_p$-nested tree if all
descendants of their inner nodes in question have been integrated
out. For the simple function above (the same that has been used in
Example \ref{ex:determinantEx1}), the joint distribution of $x_3$
and $\f_1=\|(x_1,x_2)^\top\|_{p_1}$ would be an example of such a
marginal. Since marginalization affects the $L_p$-nested tree
vertically, we call this type of marginals {\em layer marginals}. In
the following, we present their general form.

From the form of a general $L_p$-nested function and the corresponding
symmetric distribution, one might think that the layer marginals are
$L_p$-nested again. However, this is not the case since the
distribution over the $L_p$-nested unit sphere would deviate from the
uniform distribution in most cases if the distribution of its children
was $L_p$-spherically symmetric.

\begin{proposition}
  \label{pro:layermarginals}
  Let $f$ be an $L_p$-nested function. Suppose we integrate out
  complete subtrees from the tree associated with $f$, that is we
  transform subtrees into radial times uniform variables and integrate
  out the latter. Let $\mathcal J$ be the set of multi-indices of
  those nodes that have become new leaves, i.e. whose subtrees have
  been removed, and let $n_J$ be the number of leaves (in the original
  tree) in the subtree under the node $J$. Let $\x_{\widehat{\mathcal
      J}}\in\R^m$ denote those coefficients of $\x$ that are still
  part of that smaller tree and let $\ff_{\mathcal J}$ denote the
  vector of inner nodes that became new leaves. The joint distribution
  of $\x_{\widehat{ \mathcal J}}$ and $\ff_{\mathcal J}$ is given by
  \begin{align}
  \rho( \x_{\widehat{\mathcal J}},\ff_{\mathcal J}) &= \frac{\varrho(f(\x_{\widehat{\mathcal J}},\ff_{\mathcal J}))}{S_f(f(\x_{\widehat{\mathcal J}},\ff_{\mathcal J}))}\prod_{J\in\mathcal J}\f_{J}^{n_J-1}\label{eqn:layermarginal}.
  \end{align}
  
\end{proposition}
\begin{proof}The proof can be found in the Appendix \ref{app:layermarginals}.\end{proof}

Equation (\ref{eqn:layermarginal}) has an interesting special case
when considering the joint distribution of the root node's children. 

\begin{corollary}
\label{cor:layermarginals}
  The children of the root node
  $\ff_{1:\ell_\nod}=(\f_1,...,\f_{\ell_\nod})^\top$ follow the distribution
  
  \begin{align*}
    \rho(\ff_{1:\ell_\nod}) &=   \frac{p_\nod^{\ell_\nod -1} \Gamma\left[\frac{n}{p_\nod}\right]}
    {f(\f_1,...,\f_{\ell_\nod})^{n -1} 2^m \prod_{k=1}^{\ell_\nod}\Gamma\left[\frac{n_{k}}{p_\nod}\right]}
    \varrho\(f(\f_1,...,\f_{\ell_\nod})\)
    \prod_{i=1}^{\ell_\nod}\f_i^{n_i-1}
  \end{align*}
  where $m\le \ell_\nod$ is the number of leaves directly attached to
  the root node. In particular, $\ff_{1:\ell_\nod}$ can be written as
  the product $RU$, where $R$ is the $L_p$-nested radius and the
  single $|U_i|^{p_\nod}$ are Dirichlet distributed,
  i.e. $(|U_1|^{p_\nod},...,|U_{\ell_\nod}|^{p_\nod})\sim
  \text{Dir}\left[\frac{n_{1}}{p_\nod},...,\frac{n_{\ell_\nod}}{p_\nod}\right]$.
\end{corollary}
\begin{proof}
  The joint distribution is simply the application of Proposition
  (\ref{pro:layermarginals}).  Note that
  $f(\f_1,...,\f_{\ell_\nod})=||\ff_{1:\ell_\nod}||_{p_\nod}$. Applying
  the pointwise transformation $s_i = |u_i|^{p_\nod}$ yields
  $$(|U_1|^{p_\nod},...,|U_{\ell_\nod-1}|^{p_\nod})\sim
  \text{Dir}\left[\frac{n_1}{p_\nod},...,\frac{n_{\ell_\nod}}{p_\nod}\right].$$
\end{proof}

The Corollary shows that the values $f_I(\x_I)$ at inner nodes $I$, in
particular the ones directly below the root node, deviate considerably
from $L_p$-spherical symmetry. If they were $L_p$-spherically
symmetric, the $|U_i|^p$ should follow a Dirichlet distribution with
parameters $\alpha_i=\frac{1}{p}$ as has been already shown by
\cite{song:1997}. The Corollary is a generalization of their result.

We can use the Corollary to prove an interesting fact about
$L_p$-nested symmetric distributions: The only factorial $L_p$-nested
symmetric distribution must be $L_p$-spherically symmetric.

\begin{proposition}
  \label{pro:noIndep}
  Let $\x$ be $L_p$-nested symmetric distributed with independent
  marginals. Then $\x$ is $L_p$-spherically symmetric distributed. In
  particular, $\x$ follows a  $p$-generalized Normal distribution.
\end{proposition}
\begin{proof}
The proof can be found in the Appendix \ref{app:noIndep}.
\end{proof}

One immediate implication of Proposition \ref{pro:noIndep} is that
there is no factorial probability model corresponding to mixed norm
regularizers which have of the form $\sum_{i=1}^{k} \|\x_{I_k}\|_p^q$
where the index sets $I_k$ form a partition of the dimensions
$1,...,n$  \citep[see e.g.][]{zhao:2008,yuan:2006,kowalski:2008}. Many
machine learning algorithms are equivalent to minimizing the sum of a
regularizer $R(\bm w)$ and a loss function $L(\bm w,\x_1,...,\x_m)$
over the coefficient vector $\bm w$. If the $\exp\(-R(\bm w)\)$ and
$\exp\(-L(\bm w, \x_1,...,\x_m)\)$ correspond to normalizeable density
models, the minimizing solution of the objective function can be seen
as the Maximum A Posteriori (MAP) estimate of the posterior $p\(\bm w|
\x_1,...,\x_m\) \propto p(\bm w) \cdot p(\x_1,...,\x_m| \bm w) =
\exp\(-R(\bm w)\)\cdot \exp\(-L(\bm w, \x_1,...,\x_m)\)$. In that
sense, the regularizer naturally corresponds to the prior and the loss
function corresponds to the likelihood. Very often, regularizers are
specified as a norm over the coefficient vector $\bm w$ which in turn
correspond to certain priors. For example, in Ridge regression
\citep{hoerl:1962} the coefficients are regularized via $\|\bm
w\|_2^2$ which corresponds to a factorial zero mean Gaussian prior on
$\bm w$. The $L_1$-norm $\|\bm w\|_1$ in the LASSO estimator
\citep{tibshirani:1996}, again, is equivalent to a factorial
Laplacian prior on $\bm w$. Like in these two examples, regularizers
often correspond to a {\em factorial} prior.

Mixed norm regularizers naturally correspond to $L_p$-nested
distributions. Proposition \ref{pro:noIndep} shows that there is no
factorial prior that corresponds to such a regularizer. In particular,
it implies that the prior cannot be factorial between groups and
coefficients at the same time. This means that those regularizers
implicitly assume statistical dependencies between the coefficient
variables. Interestingly, for $q=1$ and $p=2$ the intuition behind
these regularizers is exactly that whole groups $I_k$ get switched on
at once, but the groups are sparse. The Proposition shows that this
might not only be due to sparseness but also due to statistical
dependencies between the coefficients within one group. The
$L_p$-nested distribution which implements independence between groups
will be further discussed below as a generalization of the
$p$-generalized Normal (see Section \ref{sec:isaica}). Note that the
marginals can be independent if the regularizer is of the form
$\sum_{i=1}^{k} \|\x_{I_k}\|_p^p$. However, in this case $p=q$ and the
$L_p$-nested function collapses to a simple $L_p$-norm which means
that the regularizer is not mixed norm.


\section{Estimation of and Sampling from $L_p$-Nested Symmetric
  Distributions}
\subsection{Maximum Likelihood Estimation}
\label{sec:fitting}
In this section, we describe procedures for maximum likelihood fitting
of $L_p$-nested symmetric distributions on data.  We provide a toolbox
online for fitting $L_p$-spherically symmetric and $L_p$-nested
symmetric distributions to data. The toolbox can be downloaded at
\url{http://www.kyb.tuebingen.mpg.de/bethge/code/}.

Depending on which parameters are to be estimated the complexity of
fitting an $L_p$-nested symmetric distribution varies. We start with
the simplest case and later continue with more complex
ones. Throughout this subsection, we assume that the model has the
form $p(\x) = \rho(Wx)\cdot |\det W| =
\frac{\varrho(W\x)}{f(W\x)^{n-1}\mathcal S_f(1))} \cdot |\det W |$
where $W\in \R^{n\times n}$ is a complete whitening matrix. This means
that given any whitening matrix $W_0$, the freedom in fitting $W$ is
to estimate an orthonormal matrix $Q\in SO(n)$ such that
$W=QW_0$. This is analogous to the case of elliptically contoured
distributions where the distributions can be endowed with $2$nd-order
correlations via $W$. In the following, we ignore the determinant of
$W$ since that data points can always be rescaled such that $\det W
=1$.

The simplest case is to fit the parameters of the radial distribution
when the tree structure, the values of the $p_I$ and $W$ are
fixed. Due to the special form of $L_p$-nested symmetric distributions
(\ref{eqn:generalNuSphericalDistribution}) it then suffices to carry
out maximum likelihood estimation on the radial component only, which
renders maximum likelihood estimation efficient and robust. This is
because the only remaining parameters are the parameters $\bm
\vartheta$ of the radial distribution and, therefore,
\begin{align*}
\text{argmax}_{\bm \vartheta} \log \rho(W\x|\bm \vartheta) 
&= \text{argmax}_{\bm
  \vartheta}\(-\log \mathcal S_f(f(W\x)) +\log \varrho(f(W\x)|\bm
\vartheta) \)\\
&= \text{argmax}_{\bm \vartheta}\log \varrho(f(W\x)|\bm \vartheta).
\end{align*}

In a slightly more complex case, when only the tree structure and $W$
are fixed, the values of the $p_I,\:I\in\mathcal I$ and $\bm
\vartheta$ can be jointly estimated via gradient ascent on the
log-likelihood. The gradient for a single data point $\x$ with respect
to the vector $\bm p$ that holds all $p_I$ for all $I \in\mathcal I$
is given by
\begin{align*}
\nabla_{\bm p}\log\rho(W\x)&=\frac{d}{dr}\log\varrho(f(W\x)) \cdot
\nabla_{\bm p} f(W\x)-\frac{(n-1)}{f(W\x)}\nabla_{\bm p}f(W\x)-\nabla_{\bm p}\log\mathcal{S}_{f}(1).
\end{align*}
For i.i.d. data points $\x_i$ the joint gradient is given by the sum
over the gradients for the single data points. Each of them involves
the gradient of $f$ as well as the gradient of the log-surface area of
$\mathds L_f$ with respect to $\bm p$, which can be computed via the
recursive equations
\begin{align*}
\frac{\partial}{\partial p_{J}}v_{I}&=\begin{cases}
0 & \mbox{ if }I\mbox{ is not a prefix of }J\\
v_{I}^{1-p_{I}}v_{I,k}^{p_{I}-1}\cdot\frac{\partial}{\partial
  p_{J}}v_{I,k} & \mbox{ if }I\mbox{ is a prefix of }J\\
\frac{v_{J}}{p_{J}}\left(v_{J}^{-p_{J}}\sum_{k=1}^{\ell_{J}}v_{J,k}^{p_{J}}\cdot\log v_{J,k}-\log v_{J}\right) & \mbox{ if }J=I\end{cases}
\end{align*}
and 
\begin{align*}
\frac{\partial}{\partial p_{J}}\log\mathcal{S}_{f}(1) =&
-\frac{\ell_J-1}{p_{J}} +
\sum_{k=1}^{\ell_{J}-1}\Psi\left[\frac{\sum_{i=1}^{k+1}n_{J,k}}{p_{J}}\right]\frac{\sum_{i=1}^{k+1}n_{J,k}}{p_{J}^{2}}
\\ & -\sum_{k=1}^{\ell_{J}-1} \Psi\left[\frac{\sum_{i=1}^{k}n_{J,k}}{p_{J}}\right]\frac{\sum_{i=1}^{k}n_{J,k}}{p_{J}^{2}}-\sum_{k=1}^{\ell_{J}-1}\Psi\left[\frac{n_{J,k+1}}{p_{J}}\right]\frac{n_{J,k+1}}{p_{J}^{2}},
\end{align*}
where $\Psi[t]=\frac{d}{dt}\log \Gamma[t]$ denotes the digamma
function. When performing the gradient ascent one needs to set $\bm 0$
as a lower bound for $\bm p$. Note that, in general, this optimization
might be a highly non-convex problem. 

On the next level of complexity, only the tree structure is fixed and
$W$ can be estimated along with the other parameters by joint
optimization of the log-likelihood with respect to $\bm p$, $\bm
\vartheta$ and $W$. Certainly, this optimization problem is also not
convex in general.  Usually, it is numerically more robust to whiten
the data first with some whitening matrix $W_0$ and perform a gradient
ascent on the special orthogonal group $SO(n)$ with respect to $Q$ for
optimizing $W=QW_0$. Given the gradient $\nabla_{W} \log\rho(W\x)$ of
the log-likelihood the optimization can be carried out by performing
line searches along geodesics as proposed by \cite{edelman:1999} (see
also \cite{absil:2007}) or by projecting $\nabla_{W} \log\rho(W\x)$ on
the tangent space $T_WSO(n))$ and performing a line search along
$SO(n)$ in that direction as proposed by \cite{manton:2002}.

The general form of the gradient to be used in such an optimization
scheme can be defined as
\begin{align*}
  &\nabla_{W}  \log\rho(W\x)  \\
  =& \nabla_{W} \(-(n-1)\cdot \log f(W\x) +\log \varrho(f(W\x))\)\\
  =&-\frac{(n-1)}{f(W\x)}\cdot \nabla_{\y}f\({W\x}\)\cdot \x^\top +  \frac{d \log
    \varrho(r)}{dr}\({f(W \x)}\) \cdot\nabla_{\y}f\({W\x}\)\cdot \x^\top.
\end{align*}
where the derivatives of $f$ with respect to $\y$ are defined by
recursive equations
\begin{align*}
\frac{\partial }{\partial y_i}\f_{I} &= 
\begin{cases}
0 &\text{ if } i \not\in I\\
\text{sgn}\:y_i &\text{ if } \f_{I,k} = |y_i|\\
\f_I^{1-p_I}\cdot\f_{I,k}^{p_{I}-1}\cdot \frac{\partial}{\partial
  y_i}\f_{I,k} &\text{ for }  i \in {I,k}.
\end{cases}
\end{align*}
Note, that $f$ might not be differentiable at $\y=0$. However, we can
always define a sub-derivative at zero, which is zero for $p_I\not=1$
and $[-1,1]$ for $p_I=1$. Again, the gradient for i.i.d. data points
$\x_i$ is given by the sum over the single gradients.

Finally, the question arises whether it is possible to estimate the
tree structure from data as well. So far, we were not able to come up
with an efficient algorithm to do so. A simple heuristic would be to
start with a very large tree, e.g. a full binary tree, and to
prune out inner nodes for which the parents and the children have
sufficiently similar values for their $p_I$. The intuition behind this
is that if they were exactly equal, they would cancel in the
$L_p$-nested function. This heuristic is certainly
sub-optimal. Firstly, the optimization will be time consuming since
there can be about as many $p_I$ as there are leaves in the
$L_p$-nested tree (a full binary tree on $n$ dimensions will have
$n-1$ inner nodes) and the repeated optimization after the pruning
steps. Secondly, the heuristic does not cover all possible trees on
$n$ leaves. For example, if two leaves are separated by the root node
in the original full binary tree, there is no way to prune out inner
nodes such that the path between those two nodes will not contain the
root node anymore.

The computational complexity for the estimation of all other
parameters despite the tree structure is difficult to assess in
general because they depend, for example, on the particular radial
distribution used. While the maximum likelihood estimation of a simple
log-Normal distribution only involves the computation of a mean and a
variance which are in $\mathcal O(m)$ for $m$ data points, a mixture
of log-Normal distributions already requires an EM algorithm which is
computationally more expensive. Additionally, the time it takes to
optimize the likelihood depends on the starting point as well as the
convergence rate and we neither have results about the convergence
rate nor is it possible to make problem independent statements about a
good initialization of the parameters. For this reason we only state
the computational complexity of single steps involved in the
optimization.

%
Computation of the gradient $\nabla_{\bm p}\log\rho(W\x)$ involves the
derivative of the radial distribution, the computation of the
gradients $\nabla_{\bm p}f(W\x)$ and $\nabla_{\bm p}\mathcal
S_f(1)$. Assuming that the derivative of the radial distribution can
be computed in $\mathcal O(1)$ for each single data point, the costly
steps are the other two gradients. Computing $\nabla_{\bm p}f(W\x)$
basically involves visiting each node of the tree once and performing
a constant number of operations for the local derivatives. Since every
inner node in an $L_p$-nested tree must have at least two children,
the worst case would be a full binary tree which has $2n-1$ nodes and
leaves. Therefore, the gradient can be computed in $\mathcal O(nm)$
for $m$ data points. For similar reasons, $f(W\x)$, $\nabla_{\bm
  p}\log\mathcal{S}_{f}(1)$ and the evaluation of the likelihood can
also be computed in $\mathcal O(nm)$. This means that each step in the
optimization of $\bm p$ can be done $\mathcal O(nm)$ plus the
computational costs for the line search in the gradient ascent.
When optimizing for $W=QW_0$ as well, the computational costs per step
increase to $\mathcal O(n^3 + n^2 m)$ since $m$ data points have to be
multiplied with $W$ at each iteration (requiring $\mathcal O(n^2 m)$
steps) and the line search involves projecting $Q$ back onto $SO(n)$
which requires an inverse matrix square root or a similar computation
in $\mathcal O(n^3)$.

For comparison, each step of fast ICA \citep{hyvarinen:1997a} for a
complete demixing matrix takes $\mathcal O(n^2 m)$ when using
hierarchical orthogonalization and $\mathcal O(n^2 m + n^3)$ for
symmetric orthogonalization. The same applies to fitting an ISA model
\citep{hyvarinen:2000,hyvarinen:2006,hyvarinen:2007}. A Gaussian
Scale Mixture (GSM) model does not need to estimate another orthogonal
rotation $Q$ because it belongs to the class of spherically symmetric
distributions and is, therefore, invariant under transformations from
$SO(n)$ \citep{wainwright:2000}. Therefore, fitting a GSM corresponds
to estimating the parameters of the scale distribution which is in
$\mathcal O(nm)$ in the best case but might be costlier depending on
the choice of the scale distribution.

\subsection{Sampling}
\label{sec:sampling}
In this section, we derive a sampling scheme for arbitrary
$L_p$-nested symmetric distributions which can for example be used for
solving integrals when using $L_p$-nested symmetric distributions for
Bayesian learning. Exact sampling from an arbitrary $L_p$-nested
symmetric distribution is in fact straightforward due to the following
observation: Since the radial and the uniform component are
independent, normalizing a sample from any $L_p$-nested distribution
to $f$-length one yields samples from the uniform distribution on the
$L_p$-nested unit sphere. By multiplying those uniform samples with
new samples from another radial distribution, one obtains samples from
another $L_p$-nested distribution. Therefore, for each $L_p$-nested
function $f$ a single $L_p$-nested distribution which can be easily
sampled from is enough. Sampling from all other $L_p$-nested
distributions with respect to $f$ is then straightforward due to the
method we just described. \cite{gupta:1997} sample from the
$p$-generalized Normal distribution since it has independent marginals
which makes sampling straightforward. Due to Proposition
\ref{pro:noIndep}, no such factorial $L_p$-nested distribution
exists. Therefore, a sampling scheme like for $L_p$-spherically
symmetric distributions is not applicable. Instead we choose to sample
from the uniform distribution inside the $L_p$-nested unit ball for
which we already computed the radial distribution in Example
\ref{ex:uniformDistributionUnitBall}.  The distribution has the form
$\rho(\x)=\frac{1}{\mathcal V_f(1)}$. In order to sample from that
distribution, we will first only consider the uniform distribution in
the positive quadrant of the unit $L_p$-nested ball which has the form
$\rho(\x)=\frac{2^n}{ \mathcal V_f(1)}$. Samples from the uniform
distributions inside the whole ball can be obtained by multiplying
each coordinate of a sample with independent samples from the uniform
distribution over $\{-1,1\}$.

The idea of the sampling scheme for the uniform distribution inside
the $L_p$-nested unit ball is based on the computation of the volume
of the $L_p$-nested unit ball in Proposition
\ref{pro:VolumeSurface}. The basic mechanism underlying the sampling
scheme below is to apply the steps of the proof backwards, which is
based on the following idea: The volume of the $L_p$-unit ball can be
computed by computing its volume on the positive quadrant only and
multiplying the result with $2^n$ afterwards. The key is now to not
transform the whole integral into radial and uniform coordinates at
once, but successively upwards in the tree. We will demonstrate this
through a little example which also should make the sampling scheme below
more intuitive. Consider the $L_p$-nested function
\begin{align*}
  f(\x)&=\left(|x_{1}|^{p_{\emptyset}} +
    \left(|x_{2}|^{p_{1}}+|x_{3}|^{p_{1}}\right)^{\frac{p_{\emptyset}}{p_{1}}}\right)^{\frac{1}{p_{\emptyset}}}.
\end{align*}
In order to solve the integral 
$$\int_{\{\x : f(\x)\le 1 \:\& \:\x \in \R_+^{n}\}}d\x,$$
we first transform $x_2$ and $x_3$ into radial and uniform coordinates
only. According to Proposition \ref{pro:DetJacobian} the determinant
of the mapping $(x_2,x_3)\mapsto (\f_1,\tilde{u}) =
(\|\x_{2:3}\|_{p_1},\x_{2:3}/\|\x_{2:3}\|_{p_1})$ is given by
$\f_1(1-\tilde{u}^{p_1})^\frac{1-p_1}{p_1}$. Therefore the integral transforms
into 
$$\int_{\{\x : f(\x)\le 1 \:\& \:\x \in \R_+^{n}\}}d\x = \int_{\{\f_1,x_1 :
  f(x_1,\f_1)\le 1 \:\& \:x_1,\f_1 \in \R_+\}}\int \int
\f_1(1-\tilde{u}^{p_1})^\frac{1-p_1}{p_1}dx_1 d\f_1 d\tilde{u} .$$
Now we can separate the integrals over $x_1$ and $\f_1$, and the
integral over $\tilde{u}$ since the boundary of the outer integral does only
depend on $\f_1$ and not on $\tilde{u}$:
$$\int_{\{\x : f(\x)\le 1 \:\& \:\x \in \R_+^{n}\}}d\x =\int (1-\tilde{u}^{p_1})^\frac{1-p_1}{p_1} d\tilde{u}\cdot \int_{\{\f_1,x_1 :
  f(x_1,\f_1)\le 1 \:\& \:x_1,\f_1 \in \R_+\}}\int \f_1dx_1 d\f_1 .$$
The value of the first integral is known explicitly since the integrand
equals the uniform distribution on the $\|\cdot\|_{p_1}$-unit
sphere. Therefore the value of the integral must be its normalization
constant which we can get using Proposition \ref{pro:VolumeSurface}.
$$\int (1-\tilde{u}^{p_1})^\frac{1-p_1}{p_1} d\tilde{u} =
\frac{\Gamma\left[\frac{1}{p_1}\right]^2\cdot p_1}{\Gamma\left[\frac{2}{p_1}\right]}.$$
An alternative way to arrive at this result is to use the
transformation $s=\tilde{u}^{p_1}$ and to notice that the integrand is a
Dirichlet distribution with parameters $\alpha_i=\frac{1}{p_1}$. The
normalization constant of the Dirichlet distribution and the constants
from the determinant Jacobian of the transformation yield the same
result. 

In order to compute the remaining integral, the same method can be
applied again yielding the volume of the $L_p$-nested unit ball. The
important part for the sampling scheme, however, is not the volume
itself but the fact that the intermediate results in this integration
process equal certain distributions. As shown in Example
\ref{ex:uniformDistributionUnitBall} the radial distribution of the
uniform distribution on the unit ball is $\beta \left[n,1\right]$, and
as just indicated by the example above the intermediate results can be
seen as transformed variables from a Dirichlet distribution. This fact
holds true even for more complex $L_p$-nested unit balls although the
parameters of the Dirichlet distribution can be slightly
different. Reversing the steps leads us to the following sampling
scheme. First, we sample from the $\beta$-distribution which gives us
the radius $\f_\nod$ on the root node. Then we sample from the
appropriate Dirichlet distribution and exponentiate the samples by
$\frac{1}{p_\nod}$ which transforms them into the analogs of the
variable $u$ from above. Scaling the result with the sample $\f_\nod$
yields the values of the root node's children, i.e. the analogs of
$x_1$ and $\f_1$. Those are the new radii for the levels below them
where we simply repeat this procedure with the appropriate Dirichlet
distributions and exponents. The single steps are summarized in
Algorithm \ref{alg:sampling}.

\begin{algorithm}
\begin{itemize}
   \item[\textbf{Input:}] The radial distribution $\varrho$ of an
   $L_p$-nested
   distribution $\rho$ for the $L_p$-nested function $f$.\\
   \item[\textbf{Output:}]  Sample $\x$ from $\rho$.\\
\end{itemize}
\textbf{Algorithm}

\begin{enumerate}
\item Sample $\f_\nod$ from a beta distribution $\beta \left[n,1\right]$.
\item For each inner node $I$ of the tree associated with $f$ sample
  the auxiliary variable $\bm s_I$ from a Dirichlet distribution
  $\text{Dir}\left[\frac{n_{I,1}}{p_I},...,\frac{n_{I,\ell_I}}{p_I}\right]$
  where $n_{I,k}$ are the number of leaves in the subtree under node
  $I,k$. Obtain coordinates on the $L_{p}$-nested sphere within the
  positive orthant by $ \bm s_I \mapsto \bm
  s_I^{\frac{1}{p_I}}=\twbu_I$ (the exponentiation is taken component-wise).
\item Transform these samples to Cartesian coordinates by
  $\f_I\cdot\twbu_I = \ff_{I,1:\ell_I}$ for each inner node, starting
  from the root node and descending to lower layers. The components of
  $\ff_{I,1:\ell_I}$ constitute the radii for the layer direct below
  them. If $I=\nod$, the radius had been sampled in step 1.
\item Once the two previous steps have been repeated until no inner
  node is left, we have a sample $\x$ from the uniform distribution in
  the positive quadrant. Normalize $\x$ to get a uniform sample from
  the sphere $\bm \u =\frac{\x}{f(\x)}$.
\item Sample a new radius $\tilde{\f}_\nod$ from the radial distribution of the target
  radial distribution $\varrho$ and obtain the sample via
  $\tilde{\x} = \tilde{\f}_\nod\cdot \u$.
\item Multiply each entry $x_i$ of $\tilde\x$ by an independent sample
  $z_i$ from the uniform distribution over $\{-1,1\}$.
\end{enumerate}
\caption{\label{alg:sampling} Exact sampling algorithm for
  $L_p$-nested distributions}
\end{algorithm}

The computational complexity of the sampling scheme is $\mathcal
O(n)$. Since the sampling procedure is like expanding the tree node by
node starting with the root, the number of inner nodes and leaves is
the total number of samples that have to be drawn from Dirichlet
distributions. Every node in an $L_p$-nested tree must at least have
two children. Therefore, the maximal number of inner nodes and leaves
is $2n-1$ for a full binary tree. Since sampling from a Dirichlet
distribution is also in $\mathcal O(n)$ the total computational
complexity for one sample is in $\mathcal O(n)$.

\section{Robust Bayesian Inference of the Location}
\label{sec:robust}
For $L_p$-spherically symmetric distributions with a location and a
scale parameter $p(\x|\bm\mu,\tau) = \tau^n
\rho(\|\tau(\x-\bm\mu)\|_p)$ \cite{osiewalski:1993} derived the
posterior in closed form using a prior $p(\bm \mu,\tau) = p(\mu) \cdot
c \cdot \tau^{-1}$, and showed that $p(\x,\bm\mu)$ does not depend on
the radial distribution $\varrho$, i.e. the particular type of
$L_p$-spherically symmetric distributions used for a fixed $p$. The
prior on $\tau$ corresponds to an improper Jeffrey's prior which is
used to represent lack of prior knowledge on the scale. The main
implication of their result is that Bayesian inference of the location
$\bm \mu$ under that prior on the scale does not depend on the
particular type of $L_p$-spherically symmetric distribution used for
inference. This means that under the assumption of an
$L_p$-spherically symmetric distributed variable, for a fixed $p$, one
does have to know the exact form of the distribution in order to
compute the location parameter.

It is straightforward to generalize their result to $L_p$-nested
symmetric distributions and, hence, making it applicable to a larger
class of distributions. Note that when using any $L_p$-nested
symmetric distribution, introducing a scale and a location via the
transformation $\x \mapsto \tau(\x-\bm\mu)$ introduces a factor of
$\tau^n$ in front of the distribution.

\begin{proposition}
  For fixed values $p_\nod,p_1,...$ and two independent priors
  $p(\bm\mu,\tau) = p(\bm \mu)\cdot c\tau^{-1}$ of the location $\mu$
  and the scale $\tau$ where the prior on $\tau$ is an improper
  Jeffrey's prior, the joint distribution $p(\x,\bm \mu)$ is given by
\begin{align*}
p(\x,\bm \mu) = f(\x - \bm \mu)^{-n} \cdot c \cdot \frac{1}{Z} \cdot p(\bm \mu),
\end{align*}
where $Z$ denotes the normalization constant of the $L_p$-nested
uniform distribution.
\end{proposition}
\begin{proof}
Given any $L_p$-nested symmetric distribution $\rho(f(\x))$ the
transformation into the polar-like coordinates yields the following
relation 
\begin{align*}
1=\int \rho(f(\x)) d\x = \int\int \prod_{L\in\mathcal L}\G_L(\bm u_{\widehat L}) r^{n-1}  
\rho(r) dr d\u =\int \prod_{L\in\mathcal L}\G_L(\bm u_{\widehat L}) d\u \cdot \int r^{n-1}  
\rho(r) dr .
\end{align*}
Since $\prod_{L\in\mathcal L}\G_L(\bm u_{\widehat L})$ is the unnormalized uniform
distribution on the $L_p$-nested unit sphere, the integral must equal
the normalization constant that we denote with $Z$ for brevity (see
Proposition \ref{pro:LpNestedUniform} for an explicit
expression). This implies that $\rho$ has to fulfill
\begin{align*}
  \frac{1}{Z}= \int r^{n-1} \rho(r) dr .
\end{align*}
Writing down  the joint distribution of $\x,\bm \mu$ and $\tau$, and
using the substitution $s = \tau f(\x-\bm\mu)$ we obtain
\begin{align*}
  p(\x,\bm\mu) &= \int \tau^n \rho(f(\tau (\x - \bm\mu))) \cdot c
  \tau^{-1} \cdot p(\bm \mu) d\tau\\
  &= \int s^{n-1} \rho(s) \cdot c \cdot p(\bm \mu) f(\x - \bm
  \mu)^{-n} ds\\
  &= f(\x - \bm \mu)^{-n} \cdot c \cdot \frac{1}{Z} \cdot p(\bm \mu).
\end{align*}
\end{proof}

Note that this result could easily be extended to $\nu$-spherical
distributions. However, in this case the normalization constant $Z$
cannot be computed for most cases and, therefore, the posterior would
not be known explicitly.

\section{Relations to ICA,  ISA and Over-Complete Linear Models}
\label{sec:isaica}
In this section, we explain the relations among $L_p$-spherically
symmetric, $L_p$-nested symmetric, ICA and ISA models. For a general
overview see Figure \ref{fig:classdiagram2}.

The density model underlying ICA models the joint distribution of the
signal $\x$ as a linear superposition of statistically independent
hidden sources $A\y=\x$ or $\y = W\x$. If the marginals of the hidden
sources are belong to the exponential power family we obtain the
$p$-generalized Normal which is a subset of the $L_p$-spherically
symmetric class. The $p$-generalized Normal distribution $p(\y)
\propto \exp(-\tau \|\y\|_p^p)$ is a density model that is often used
in ICA algorithms for kurtotic natural signals like images and sound
by optimizing a demixing matrix $W$ w.r.t. to the model $p(\y) \propto
\exp(-\tau \|W\x\|_p^p)$
\citep{twlee:2000,zhang:2004,lewicki:2002}. It can be shown that the
$p$-generalized Normal is the only factorial model in the class of
$L_p$-spherically symmetric models \citep{sinz:2009a}, and, by
Proposition \ref{pro:noIndep}, also the only factorial $L_p$-nested
symmetric distribution.

\begin{figure}[!ht]
  \begin{center}
    \includegraphics[width=12cm]{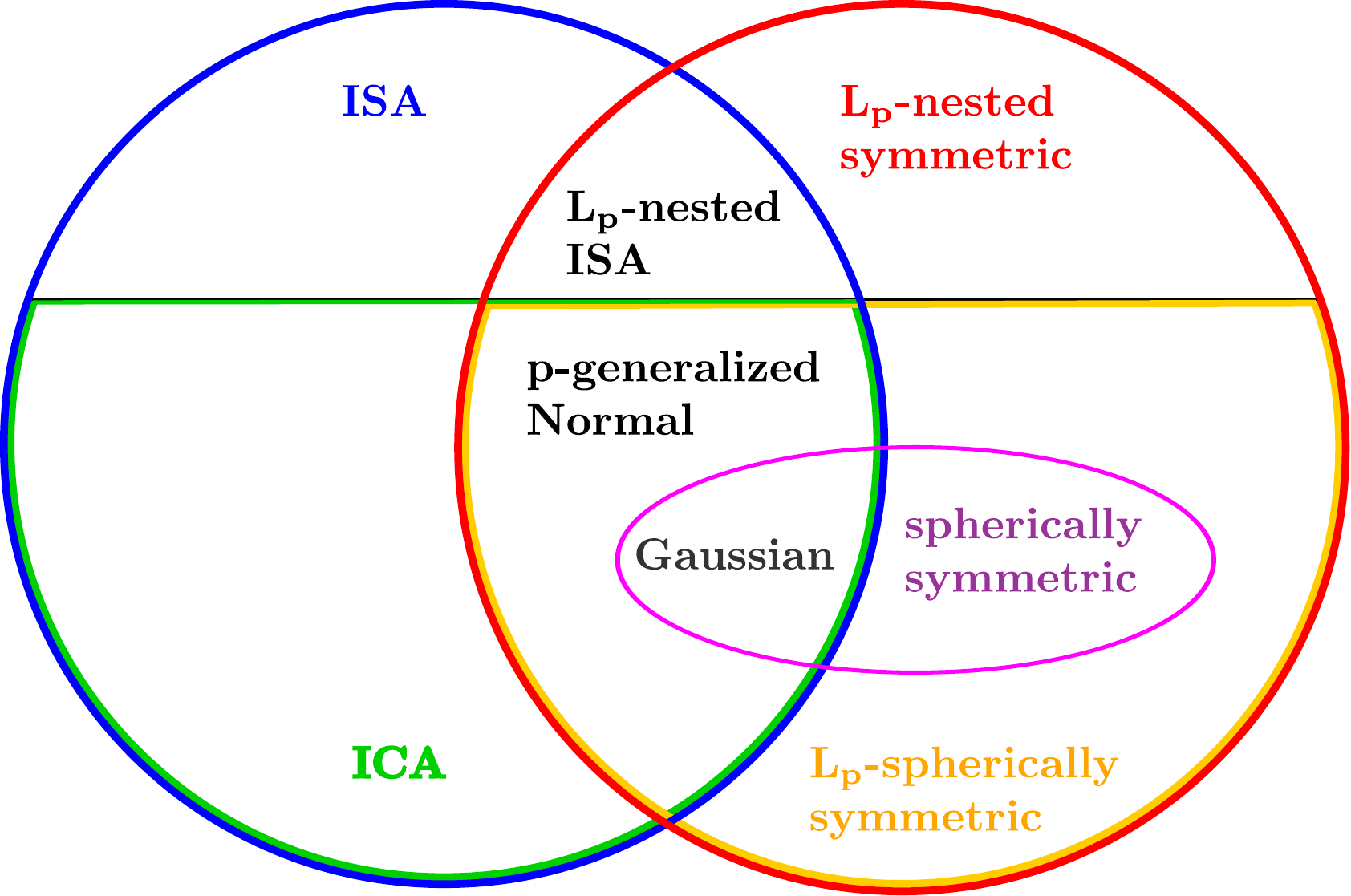}
  \end{center}
  \caption{\label{fig:classdiagram2} Relations between the different
    classes of distributions: For arbitrary distributions on the
    subspaces ISA (blue) is a superclass of ICA (green). Obviously,
    $L_p$-nested symmetric distributions (red) are a superclass of
    $L_p$-spherically symmetric distributions (yellow). $L_p$-nested
    ISA models live in the intersection of $L_p$-nested symmetric
    distributions and ISA models (intersection of red and blue). Those
    $L_p$-nested ISA models that are $L_p$-spherically symmetric are
    also ICA models (intersection of green and yellow). This is the
    class of $p$-generalized Normal distributions. If $p$ is fixed to
    two, one obtains the spherically symmetric distributions
    (pink). The only class of distributions in the intersection
    between spherically symmetric distributions and ICA models is the
    Gaussian (intersection green, yellow and pink). }
\end{figure}

An important generalization of ICA is the Independent Subspace
Analysis (ISA) proposed by \cite{hyvarinen:2000} and by
\cite{hyvarinen:2007} who used $L_p$-spherically symmetric
distributions to model the single subspaces. Like in ICA, also ISA
models the hidden sources of the signal as a product of multivariate
distributions:
$$p(\y) = \prod_{k=1}^K p(\y_{I_k}).$$ Here, $\y = W\x$
and $I_k$ are index sets selecting the different subspaces from the
responses of $W$ to $\x$. The collection of index sets $I_k$ forms a
partition of $1,...,n$. ICA is a special case of ISA in which
$I_k=\{k\}$ such that all subspaces are one-dimensional. For the ISA
models used by Hyv{\"a}rinen et al. the distribution on the subspaces
was chosen to be either spherically or $L_p$-spherically symmetric.

ICA and ISA have been used to infer features from natural signals, in
particular from natural images. However, as mentioned by several
authors \citep{zetzsche:1993,simoncelli:1997,wainwright:2000} and
demonstrated quantitatively by \cite{bethge:2006} and
\cite{eichhorn:2008}, the assumptions underlying linear ICA are not
well matched by the statistics of natural images. Although the
marginals can be well described by an exponential power family, the
joint distribution cannot be factorized with linear filters $W$.

A reliable parametric way to assess how well the independence
assumption is met by a signal at hand is to fit a more general class
of distributions that contains factorial as well as non-factorial
distributions which both can equally well reproduce the marginals. By
comparing the likelihood on held out test data between the best
fitting non-factorial and the best-fitting factorial case, one can
asses how well the sources can be described by a factorial
distribution. For natural images, for example, one can use an
arbitrary $L_p$-spherically symmetric distribution $\rho(\|W\x\|_p)$,
fit it to the whitened data and compare its likelihood on held out
test data to the one of the $p$-generalized Normal
\citep{sinz:2009}. Since any choice of radial distribution $\varrho$
determines a particular $L_p$-spherically symmetric distribution the
idea is to explore the space between factorial and non-factorial
models by using a very flexible density $\varrho$ on the radius. Note
that having an explicit expression of the normalization constant
allows for particularly reliable model comparisons via the
likelihood. For many graphical models, for instance, such an explicit
and computable expression is often not available.

\begin{figure}[!ht]
  \begin{center}
    \includegraphics[width=12cm]{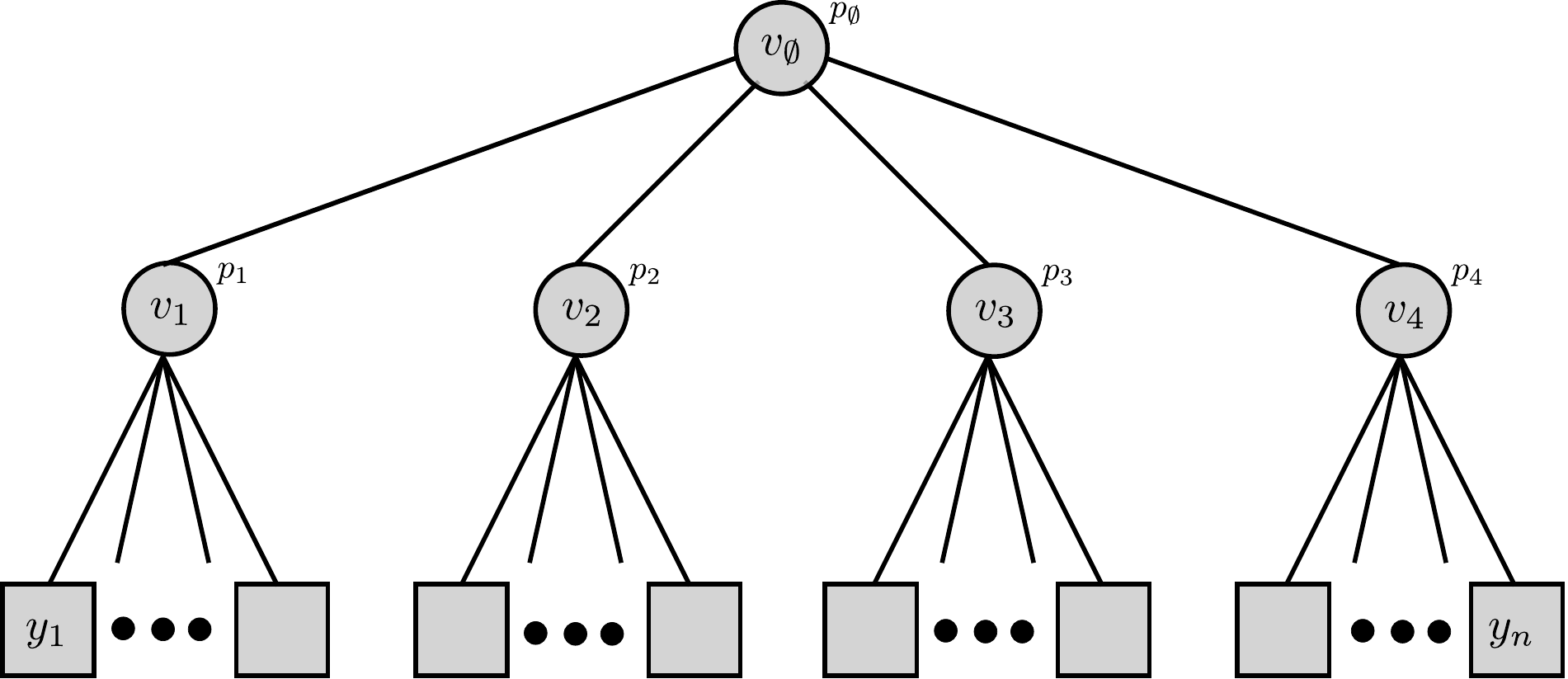}
  \end{center}
  \caption{\label{fig:ISAtree} Tree corresponding to an $L_p$-nested
    ISA model. }
\end{figure}

The same type of dependency-analysis can be carried out for ISA using
$L_p$-nested symmetric distributions \citep{sinz:2010}. Figure
\ref{fig:ISAtree} shows the $L_p$-nested tree corresponding to an ISA
with four subspaces. For general such trees, each inner node---except
the root node---corresponds to a single subspace. When using the
radial distribution
\begin{align}
\varrho_\nod(\f_\nod)= \frac{p_\nod\f_\nod^{n-1}}{\Gamma\left[\frac{n}{p_\nod}\right]s^\frac{n}{p_\nod}}\exp\(-\frac{\f_\nod^{p_\nod}}{s}\)
\label{eq:pndISAradial}
\end{align}
the subspaces $\f_1,...,\f_{\ell_\nod}$ become independent and one
obtains an ISA model of the form
\begin{align*}
  \rho(\y) &= \frac{1}{Z} \exp\(-\frac{f(\y)^{p_\nod}}{s}\)\\
  &= \frac{1}{Z} \exp\(-\frac{\sum_{k=1}^{\ell_\nod} \|\y_{I_k}\|_{p_k}}{s}\)\\
  &=
  \frac{p^{\ell_\nod}_\nod}{s^{\frac{n}{p_\nod}}\prod_{i=1}^{\ell_\nod}\Gamma\left[\frac{n_i}{p_\nod}\right]}\exp\(-\frac{\sum_{k=1}^{\ell_\nod}
    \|\y_{I_k}\|_{p_k}}{s}\)
  \prod_{k=1}^{\ell_\nod}\frac{p_k^{\ell_k-1}\Gamma\left[\frac{n_k}{p_k}\right]}{2^{n_k}\Gamma^{n_k}\left[\frac{1}{p_I}\right]},
\end{align*}
which has $L_p$-spherically symmetric distributions on each
subspace. Note that this radial distribution is equivalent to a Gamma
distribution whose variables have been raised to the power of
$\frac{1}{p_\nod}$. In the following we will denote distributions of
this type with $\gamma_p\(u,s\)$, where $u$ and $s$ are the shape and
scale parameter of the Gamma distribution, respectively. The
particular $\gamma_p$ distribution that results in independent
subspaces has arbitrary scale but shape parameter
$u=\frac{n}{p_\nod}$. When using any other radial distribution, the
different subspaces do not factorize and the distribution is also not
an ISA model. In that sense $L_p$-nested symmetric distributions are a
generalization of ISA. Note, however, that not every ISA model is also
$L_p$-nested symmetric since not every product of arbitrary
distributions on the subspaces, even if they are $L_p$-spherically
symmetric, must also be $L_p$-nested.

It is natural to ask, whether $L_p$-nested symmetric distributions can
serve as a prior distribution $p(\y|\bm \vartheta)$ over hidden
factors in over-complete linear models of the form $p(\x|W,\sigma,\bm
\vartheta) = \int p(\x|W\y,\sigma) p(\y|\bm \vartheta) d\y$, where
$p(\x|W\y)$ represents the likelihood of the observed data point $\x$
given the hidden factors $\y$ and the over-complete matrix $W$. For
example, $p(\x|W\y,\sigma) = \mathcal N(W\y,\sigma\cdot I)$ could be a
Gaussian like in \cite{olshausen:1996}. Unfortunately, such a model
would suffer from the same problems as all over-complete linear
models: While sampling from the prior is straightforward sampling from
the posterior $p(\y|\x,W,\bm \vartheta, \sigma)$ is difficult because
a whole subspace of $\y$ leads to the same $\x$. Since parameter
estimation either involves solving the high-dimensional integral
$p(\x|W,\sigma,\bm \vartheta) = \int p(\x|W\y,\sigma) p(\y|\bm
\vartheta) d\y$ or sampling from the posterior, learning is
computationally demanding in such models. Various methods have been
proposed to learn $W$, ranging from sampling the posterior only at its
maximum \citep{olshausen:1996}, approximating the posterior with a
Gaussian via the Laplace approximation \citep{lewicki:1999} or using
Expectation Propagation \citep{seeger:2008}. In particular, all of
the above studies either do not fit hyper-parameters $\bm \vartheta$
for the prior \citep{olshausen:1996,lewicki:1999} or rely on the
factorial structure of it \citep{seeger:2008}. Since $L_p$-nested
distributions do not provide such a factorial prior, Expectation
Propagation is not directly applicable. An approximation like in
\cite{lewicki:1999} might be possible, but additionally estimating the
parameters $\bm \vartheta$ of the $L_p$-nested symmetric distribution
adds another level of complexity in the estimation
procedure. Exploring such over-complete linear models with a
non-factorial prior may be an interesting direction to investigate,
but it will need a significant amount of additional numerical and
algorithmical work to find an efficient and robust estimation
procedure.

\section{Nested Radial Factorization with $L_p$-Nested Symmetric
  Distributions}
\label{sec:nonlinearICA}
$L_p$-nested symmetric distribution also give rise to a non-linear ICA
algorithm for linearly mixed non-factorial $L_p$-nested hidden sources
$\y$. The idea is similar to the Radial Factorization algorithms
proposed by \cite{lyu:2009} and \cite{sinz:2009}. For this reason, we
call it {\em Nested Radial Factorization (NRF)}. For a one layer
$L_p$-nested tree, NRF is equivalent to Radial Factorization as
described in \cite{sinz:2009}. If additionally $p$ is set to $p=2$,
one obtains the Radial Gaussianization by \cite{lyu:2009}. Therefore,
NRF is a generalization of Radial Factorization. It has been
demonstrated that Radial Factorization algorithms outperform linear
ICA on natural image patches \citep{lyu:2009,sinz:2009}. Since
$L_p$-nested symmetric distributions are slightly better in likelihood
on natural image patches \citep{sinz:2010} and since the difference
in the average log-likelihood directly corresponds to the reduction in
dependencies between the single variables \citep{sinz:2009}, NRF will
slightly outperform Radial Factorization on natural images. For other
type of data the performance will depend on how well the hidden
sources can be modeled by a linear superposition of---possibly
non-independent---$L_p$-nested symmetrically distributed sources. Here
we state the algorithm as a possible application of
$L_p$-nested symmetric distributions for unsupervised learning.

The idea is based on the observation that the choice of the radial
distribution $\varrho$ already determines the type of $L_p$-nested
symmetric distribution. This also means that by changing the radial
distribution by remapping the data, the distribution could possibly be
turned in a factorial one. Radial Factorization algorithms fit an
$L_p$-spherically symmetric distribution with a very flexible radial
distribution to the data and map this radial distribution
$\varrho_{s}$ ($s$ for \underline{s}ource) into the one of a
$p$-generalized Normal distribution by the mapping
\begin{align}
\y \mapsto \frac{(\mathcal F^{-1}_{\bot\!\!\!\bot}\circ \mathcal F_{s})(\|\y\|_p)}{\|\y\|_p}\cdot \y,\label{eq:nonlinICAMap}
\end{align}
where $\mathcal F_{\bot\!\!\!\bot}$ and $\mathcal F_{s}$ are the
cumulative distribution functions of the two radial distributions
involved. The mapping basically normalizes the demixed source $\y$ and
rescales it with a new radius that has the correct distribution.

Exactly the same method cannot work for $L_p$-nested symmetric
distributions since Proposition \ref{pro:noIndep} states that there is
no factorial distribution we could map the data to by merely changing
the radial distribution. Instead we have to remap the data in an
iterative fashion beginning with changing the radial distribution at
the root node into the radial distribution of the $L_p$-nested ISA
shown in equation (\ref{eq:pndISAradial}). Once the nodes are
independent, we repeat this procedure for each of the child nodes
independently, then for their child nodes and so on, until only leaves
are left. The rescaling of the radii is a non-linear mapping since the
transform in equation \eqref{eq:nonlinICAMap} is
non-linear. Therefore, NRF is a non-linear ICA algorithm.

\begin{figure}[!ht]
  \begin{center}
    \includegraphics[width=12cm]{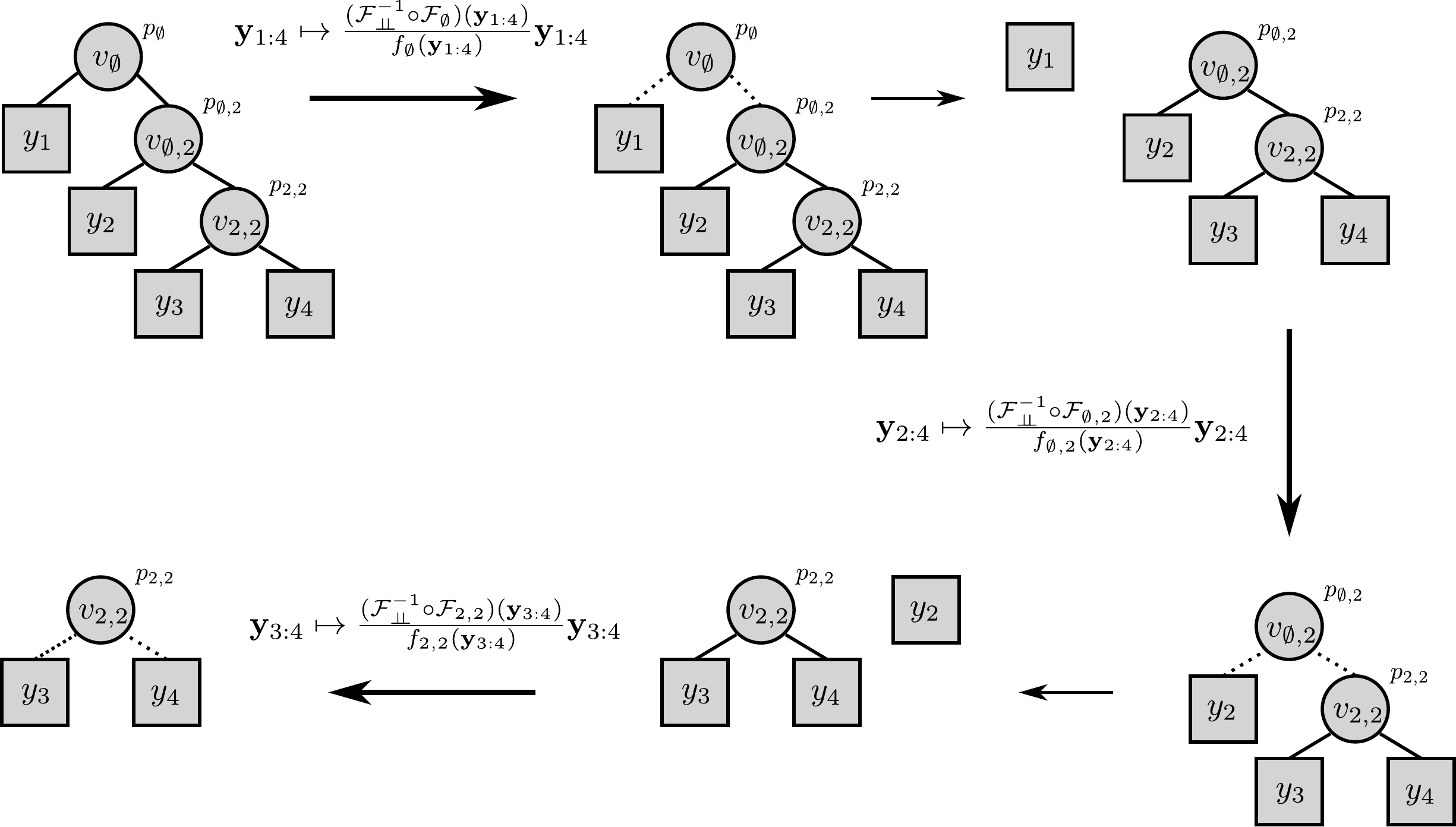}
  \end{center}
  \caption{\label{fig:nonlinearICAExample} $L_p$-nested non-linear ICA
    for the tree of Example \ref{ex:nonlinearICA}: For an arbitrary
    $L_p$-nested symmetric distribution, using equation
    (\ref{eq:nonlinICAMap}), the radial distribution can be remapped
    such that the children of the root node become independent. This
    is indicated in the plot via dotted lines. Once the data has been
    rescaled with that mapping, the children of root node can be
    separated. The remaining subtrees are again $L_p$-nested symmetric
    and have a particular radial distribution that can be remapped
    into the same one that makes their root nodes' children
    independent. This procedure is repeated until only leaves are
    left. }
\end{figure}

We demonstrate this at a simple example. 
\begin{example}
\label{ex:nonlinearICA}
Consider the function
\begin{align*}
f(\y)&=\left(|y_{1}|^{p_{\emptyset}} + \left(|y_{2}|^{p_{{\nod,2}}}+  \( |y_3|^{p_{2,2}} + |y_{4}|^{p_{{2,2}}}\)^\frac{p_{\nod,2}}{p_{2,2}} \right)^{\frac{p_{\emptyset}}{p_{{\nod,2}}}}\right)^{\frac{1}{p_{\emptyset}}}
\end{align*}
for $\y=W\x$ where $W$ as been estimated by fitting an $L_p$-nested
symmetric distribution with a flexible radial distribution to $W\x$ as
described in Section \ref{sec:fitting}. Assume that the data has
already been transformed once with the mapping of equation
\eqref{eq:nonlinICAMap}. This means that the current radial
distribution is given by (\ref{eq:pndISAradial}) where we chose $s=1$
for convenience. This yields a distribution of the form
\begin{align*}
  \rho(\y) &=
  \frac{p_\nod}{\Gamma\left[\frac{n}{p_\nod}\right]}\exp\(-|y_{1}|^{p_{\emptyset}} - \left(|y_{2}|^{p_{{\nod,2}}}+  \( |y_3|^{p_{2,2}} + |y_{4}|^{p_{{2,2}}}\)^\frac{p_{\nod,2}}{p_{2,2}} \right)^{\frac{p_{\emptyset}}{p_{{\nod,2}}}}\)\\
  &\times\frac{1}{2^n }\prod_{I\in \mathcal
    I}p_I^{\ell_I-1}\frac{\Gamma\left[\frac{n_I}{p_I}\right]}{\prod_{k=1}^{\ell_I}\Gamma\left[\frac{n_{I,k}}{p_I}\right]}.
\end{align*}
Now we can separate the distribution of $y_1$ from the distribution
over $y_2,...,y_4$. The distribution of $y_1$ is a $p$-generalized
Normal 
$$p(y_1) = \frac{p_\nod }{2\Gamma\left[\frac{1}{p_\nod}\right]} \exp\(-|y_1|^{p_\nod}\).$$
Thus the distribution of $y_2,...,y_4$ is given by
\begin{align*}
  \rho(y_2,...,y_4)
  &=\frac{p_\nod}{\Gamma\left[\frac{n_{\nod,2}}{p_\nod}\right]}
  \exp\(- \left(|y_{2}|^{p_{{\nod,2}}}+ \( |y_3|^{p_{2,2}} +
    |y_{4}|^{p_{{2,2}}}\)^\frac{p_{\nod,2}}{p_{2,2}}
  \right)^{\frac{p_{\emptyset}}{p_{{\nod,2}}}}\) \\
&\times\frac{1}{2^{n-1}
  }\prod_{I\in \mathcal I \backslash
    \nod}p_I^{\ell_I-1}\frac{\Gamma\left[\frac{n_I}{p_I}\right]}{\prod_{k=1}^{\ell_I}\Gamma\left[\frac{n_{I,k}}{p_I}\right]}.
\end{align*}
By using equation (\ref{eq:generalLpNested}) we can identify the new
radial distribution to be
$$\varrho(\f_{\nod,2}) = \frac{p_\nod \f_{\nod,2}^{n-2}}{\Gamma\left[\frac{n_{\nod,2}}{p_\nod}\right]}
\exp\(- \f_{\nod,2}^{p_\nod}\). $$ Replacing this distribution by the one for
the $p$-generalized Normal (for data we would use the mapping in
equation (\ref{eq:nonlinICAMap})), we obtain
\begin{align*}
  \rho(y_2,...,y_4)
  &=\frac{p_{\nod,2}}{\Gamma\left[\frac{n_{\nod,2}}{p_{\nod,2}}\right]}
  \exp\(- |y_{2}|^{p_{{\nod,2}}}- \( |y_3|^{p_{2,2}} +
    |y_{4}|^{p_{{2,2}}}\)^\frac{p_{\nod,2}}{p_{2,2}}\) \\
  &\times \frac{1}{2^{n-1}
  }\prod_{I\in \mathcal I \backslash
    \nod}p_I^{\ell_I-1}\frac{\Gamma\left[\frac{n_I}{p_I}\right]}{\prod_{k=1}^{\ell_I}\Gamma\left[\frac{n_{I,k}}{p_I}\right]}.
\end{align*}
Now, we can separate out the distribution of $y_2$ which is again
$p$-generalized Normal. This leaves us with the distribution for $y_3$
and $y_4$
\begin{align*}
  \rho(y_3,y_4)
  &=\frac{p_{{\nod,2}}}{\Gamma\left[\frac{n_{2,2}}{p_{\nod,2}}\right]}
  \exp\(- \( |y_3|^{p_{2,2}} + |y_{4}|^{p_{{2,2}}}\)^\frac{p_{\nod,2}}{p_{2,2}}\)
  \frac{1}{2^{n-2} }\prod_{I\in \mathcal I \backslash
    \{\nod,(\nod,2)\}}p_I^{\ell_I-1}\frac{\Gamma\left[\frac{n_I}{p_I}\right]}{\prod_{k=1}^{\ell_I}\Gamma\left[\frac{n_{I,k}}{p_I}\right]}.
\end{align*}
For this distribution we can repeat the same procedure which will also
yield $p$-generalized Normal distributions for $y_3$ and $y_4$.
\end{example}

This non-linear procedure naturally carries over to arbitrary
$L_p$-nested trees and distributions, thus yielding a general
non-linear ICA algorithm for linearly mixed non-factorial $L_p$-nested
sources. For generalizing Example \ref{ex:nonlinearICA}, note the
particular form of the radial distributions involved. As already noted
above the distribution (\ref{eq:pndISAradial}) on the root node's
values that makes its children statistical independent is that of a
Gamma distributed variable with shape parameter
$\frac{n_\nod}{p_\nod}$ and scale parameter $s$ which has been raised
to the power of $\frac{1}{p_\nod}$. In Section \ref{sec:isaica} we
denoted this class of distributions with $\gamma_p\left[u,s\right]$,
where $u$ and $s$ are the shape and the scale parameter,
respectively. Interestingly, the radial distributions of the root
node's children are also $\gamma_p$ except that the shape parameter is
$\frac{n_{\nod,i}}{p_\nod}$. The goal of the radial remapping of the
children's values is hence just changing the shape parameter from
$\frac{n_{\nod,i}}{p_\nod}$ to $\frac{n_{\nod,i}}{p_{\nod,i}}$. Of
course, it is also possible to change the scale parameter of the
single distributions during the radial remappings. This will not
affect the statistical independence of the resulting variables. In the
general algorithm, that we describe now, we choose $s$ such that the
transformed data is white.

The algorithm starts with fitting a general $L_p$-nested model of the
form $\rho(W \x)$ as described in Section \ref{sec:fitting}. Once this
is done, the linear demixing matrix $W$ is fixed and the hidden
non-factorial sources are recovered via $\y=W\x$. Afterwards, the
sources $\y$ are non-linearly made independent by calling the
recursion specified in Algorithm \ref{alg:non-linearICA} with the
parameters $W\x$, $f$ and $\varrho$, where $\varrho$ is the radial
distribution of the estimated model.

\begin{algorithm}
\begin{itemize}
  \item[\textbf{Input:}] Data point $\y$, $L_p$-nested function $f$, current
  radial distribution $\varrho_s$,
  \item[\textbf{Output:}] Non-linearly transformed data point $\y$\\
\end{itemize}
\textbf{Algorithm}
\begin{enumerate}
\item Set the target radial distribution to be
  $\varrho_{\bot\!\!\!\bot} \leftarrow\gamma_p\(\frac{n_\nod}{p_\nod},
  \frac{\Gamma\left[\frac{1}{p_\nod}\right]^\frac{p_\nod}{2}}{\Gamma\left[\frac{3}{p_\nod}\right]^\frac{p_\nod}{2}}\)$
\item Set $\y \leftarrow \frac{\mathcal F_{\bot\!\!\!\bot}^{-1}\(\mathcal
    F_{s}\(f(\y)\)\)}{f(\y)} \cdot \y$ where $\mathcal
  F$ denotes the cumulative distribution function the respective
  $\varrho$.
\item For all children $i$ of the root node that are not leaves:
  \begin{enumerate} 
    \item Set $\varrho_s \leftarrow\gamma_p\(\frac{n_{\nod,i}}{p_\nod},
  \frac{\Gamma\left[\frac{1}{p_\nod}\right]^\frac{p_\nod}{2}}{\Gamma\left[\frac{3}{p_\nod}\right]^\frac{p_\nod}{2}}\)$
  \item Set $\y_{\nod,i}\leftarrow$ NRF($\y_{\nod,i},f_{\nod,i},\varrho_s$). Note that in
    the recursion ${\nod,i}$ will become the new $\nod$.
  \end{enumerate}
\item Return $\y$
\end{enumerate}
\caption{\label{alg:non-linearICA} Recursion NRF($\y,f,\varrho_s$)}
\end{algorithm}

The computational complexity for transforming a single data point is
$\mathcal O(n^2)$ because of the matrix multiplication $W\x$. In the
non-linear transformation, each single data dimension is not rescaled
more that $n$ times which means that the rescaling is certainly also
in $\mathcal O(n^2)$.

An important aspect of NRF is that it yields a probabilistic model for
the transformed data. This model is simply a product of $n$
independent exponential power marginals. Since the radial remappings
do not change the likelihood, the likelihood of the non-linearly
separated data is the same as the likelihood of the data under
$L_p$-nested symmetric distribution that was fitted to it in the first
place. However, in some cases, one might like to fit a different
distribution to the outcome of Algorithm \ref{alg:non-linearICA}. In
that case the determinant of the transformation is necessary to
determine the likelihood of the input data---and not the transformed
one---under the model. The following lemma provides the determinant of
the Jacobian for the non-linear rescaling.

\begin{lemma}[Determinant of the Jacobian]
  \label{lem:radialDeterminant}
  Let $\z = $ NRF$(W\x,f,\varrho_s)$ as described above. Let
  $\t_I$ denote the value of $W\x$ below the inner node $I$ which have
  been transformed with Algorithm \ref{alg:non-linearICA} up to node
  $I$. Let $g_I(r) = (\mathcal F_{\varrho_{\bot\!\!\!\bot}}\circ
  \mathcal F_{\varrho_{s}})(r)$ denote the radial transform at node
  $I$ in Algorithm \ref{alg:non-linearICA}. Furthermore, let $\mathcal
  I$ denote the set of all inner nodes, excluding the leaves. Then, the
  determinant of the Jacobian $\(\frac{\partial z_i}{\partial
    x_j}\)_{ij}$ is given by
  \begin{align}
    \left|\det \frac{\partial z_i}{\partial x_j}\right| &=
    |\det W| \cdot \prod_{I\in \mathcal I} \left|\frac{g_I(f_I(\t_I))^{n_I-1}}{f_I(\t_I)^{n_I-1}} \cdot \frac{\varrho_{s}(f_I(\t_I))}{\varrho_{\bot\!\!\!\bot}(g_I(f_I(\t_I)))}\right|\label{eq:radialDeterminant}.
  \end{align}
\end{lemma}
\begin{proof}The proof can be found in the Appendix \ref{app:radialJacobian}.
\end{proof}

\section{Conclusion}

\label{sec:conclusion}
In this article we presented a formal treatment of the first tractable
subclass of $\nu$-spherical distributions which generalizes the
important family of $L_p$-spherically symmetric distributions. We
derived an analytical expression for the normalization constant,
introduced a coordinate system particularly tailored to $L_p$-nested
functions and computed the determinant of the Jacobian for the
corresponding coordinate transformation. Using these results, we
introduced the uniform distribution on the $L_p$-nested unit sphere
and the general form of an $L_p$-nested distribution for arbitrary
$L_p$-nested functions and radial distributions. We also derived an
expression for the joint distribution of inner nodes of an
$L_p$-nested tree and derived a sampling scheme for an arbitrary
$L_p$-nested distribution.

$L_p$-nested symmetric distributions naturally provide the class of
probability distributions corresponding to mixed norm priors, allowing
full Bayesian inference in the corresponding probabilistic models. We
showed that a robustness result for Bayesian inference of the location
parameter known for $L_p$-spherically symmetric distributions carries
over to the $L_p$-nested symmetric class. We discussed the relations
of $L_p$-nested symmetric distributions to Indepedent Component (ICA)
and Independent Subspace Analysis (ISA), and discussed its
applicability as a prior distribution in over-complete linear
models. Finally, we showed how $L_p$-nested distributions can be used
to construct a non-linear ICA algorithm called Nested Radial
Factorization (NRF). 

The application of $L_p$-nested symmetric distribution has been
presented in a previous conference paper \citep{sinz:2010}. Code for
training this class of distribution is provided online under
\url{http://www.kyb.tuebingen.mpg.de/bethge/code/}.

\section*{Acknowledgements}
We would like to thank Eero Simoncelli for mentioning the idea of
replacing $L_p$-norms by $L_p$-nested functions to us. Furthermore, we
want to thank Sebastian Gerwinn, Suvrit Sra and Reshad Hosseini for
fruitful discussions and feedback on the manuscript. Finally, we would
like to thank the anonymous reviewers for their comments that helped
to improve the manuscript.

This work is supported by the German Ministry of Education, Science,
Research and Technology through the Bernstein prize to MB (BMBF; FKZ:
01GQ0601), a scholarship to FS by the German National Academic
Foundation, and the Max Planck Society.

\appendix
\section{Determinant of the Jacobian}
\label{app:determinantOfJacobian}
\begin{proof}[Lemma \ref{lem:generalDeterminant}]
  \label{pro:generalDeterminant}
  The proof is very similar to the one in \cite{song:1997}. To derive
  equation (\ref{eq:generalDeterminant}) one needs to expand the
  Jacobian of the inverse coordinate transformation with respect to
  the last column using the Laplace expansion of the determinant. The
  term $\Delta_n$ can be factored out of the determinant and cancels
  due to the absolute value around it. Therefore, the determinant of
  the coordinate transformation does not depend on $\Delta_n$.

  The partial derivatives of the inverse coordinate transformation are
  given by:
  \begin{align*}
    \frac{\partial}{\partial u_{k}}x_{i} &=  \delta_{ik}r\mbox{ for }1\le i,k\le n-1\\
    \frac{\partial}{\partial u_{k}}x_{n} &=  \Delta_n r\frac{\partial u_{n}}{\partial u_{k}}\mbox{ for }1\le k\le n-1\\
    \frac{\partial}{\partial r}x_{i} &=  u_{i}\mbox{ for }1\le i\le n-1\\
    \frac{\partial}{\partial r}x_{n} &=  \Delta_n u_{n}.
  \end{align*}
  Therefore, the structure of the Jacobian is given by
  \begin{align*}
    \mathcal{J} &= \(\begin{array}{cccc}
        r & \dots & 0 & u_{1}\\
        \vdots & \ddots & \vdots & \vdots\\
        0 & \dots & r & u_{n-1}\\
        \Delta_n r\frac{\partial u_{n}}{\partial u_{1}} & \dots & \Delta_n
        r\frac{\partial u_{n}}{\partial u_{n-1}} & \Delta_n
        u_{n}\end{array}\).
  \end{align*}
  Since we are only interested in the absolute value of the
  determinant and since $\Delta_n\in\{-1,1\}$, we can factor out
  $\Delta_n$ and drop it. Furthermore, we can factor out $r$ from the
  first $n-1$ columns which yields
  \begin{align*}
    |\det\mathcal{J}| &=   r^{n-1}\left|\det\(\begin{array}{cccc}
          1 & \dots & 0 & u_{1}\\
          \vdots & \ddots & \vdots & \vdots\\
          0 & \dots & 1 & u_{n-1}\\
          \frac{\partial u_{n}}{\partial u_{1}} & \dots &
          \frac{\partial u_{n}}{\partial u_{n-1}} &
          u_{n}\end{array}\)\right|.
  \end{align*}
  Now we can use the Laplace expansion of the determinant with respect
  to the last column. For that purpose, let $\mathcal{J}_{i}$ denote
  the matrix which is obtained by deleting the last column and the $i$th
  row from $\mathcal{J}$. This matrix has the following structure
  \begin{align*}
    \mathcal{J}_{i} &=  \(\begin{array}{ccccccc}
        1 &  &  & 0\\
        & \ddots &  & 0\\
        &  & 1 & 0\\
        &  &  & \vdots & 1\\
        &  &  & 0 &  & \ddots\\
        &  &  & 0 &  &  & 1\\
        \frac{\partial u_{n}}{\partial u_{1}} &  &  & \frac{\partial
          u_{n}}{\partial u_{i}} &  &  & \frac{\partial u_{n}}{\partial
          u_{n-1}}\end{array}\).
  \end{align*}
  We can transform $\mathcal{J}_{i}$ into a lower triangular matrix by
  moving the column with all zeros and $\frac{\partial u_{n}}{\partial
    u_{i}}$ bottom entry to the rightmost column of
  $\mathcal{J}_{i}$. Each swapping of two columns introduces a factor
  of $-1$. In the end, we can compute the value of
  $\det\mathcal{J}_{i}$ by simply taking the product of the diagonal
  entries and obtain $\det\mathcal{J}_{i} =  (-1)^{n-1-i} \frac{\partial u_{n}}{\partial
    u_{i}}$.
  This yields
  \begin{align*}
    |\det\mathcal{J}| &=  r^{n-1}\(\sum_{k=1}^{n}(-1)^{n+k}u_{k}\det\mathcal{J}_{k}\)\\
    &=   r^{n-1}\(\sum_{k=1}^{n-1}(-1)^{n+k}u_{k}\det\mathcal{J}_{k}+(-1)^{2n}\frac{\partial x_{n}}{\partial r}\)\\
    &=  r^{n-1}\(\sum_{k=1}^{n-1}(-1)^{n+k}u_{k}(-1)^{n-1-k}\frac{\partial u_{n}}{\partial u_{k}}+u_{n}\)\\
    &=  r^{n-1}\(-\sum_{k=1}^{n-1}u_{k}\frac{\partial
        u_{n}}{\partial u_{k}}+u_{n}\).
  \end{align*}
\end{proof}

Before proving Proposition \ref{pro:DetJacobian} stating that the
determinant only depends on the terms $\G_I(\bm u_{\widehat I})$ produced by the chain rule when
used upwards in the tree, let us quickly outline the essential
mechanism when taking the chain rule for $\frac{\partial u_n}{\partial
  u_q}$: Consider the tree corresponding to $f$. By definition $u_{n}$
is the rightmost leaf of the tree. Let $L,\ell_L$ be the multi-index
of $u_n$.  As in the example, the chain rule starts at the leaf
$u_{n}$ ascends in the the tree until it reaches the lowest node whose
subtree contains both, $u_n$ and $u_q$. At this point, it starts
descending the tree until it reaches the leaf $u_{q}$.  Depending on
whether the chain rule ascends or descends, two different forms of
derivatives occur: while ascending, the chain rule produces $\G_I(\bm u_{\widehat I})$-terms
like the one in the example above. At descending, it produces
$\F_I(\u_{I})$-terms. The general definitions of the $\G_I(\bm u_{\widehat I})$- and $\F_I(\u_{I})$-terms are
given by the recursive formulae
\begin{align}
  \G_{I,\ell_I}(\bm u_{\widehat{I,\ell_I}}) &=  \g_{I,{\ell_I}}(\u_{\widehat{I,\ell_I}})^{p_{I,{\ell_I}}-p_{I}}  =\left(\g_{I}(\u_{\widehat{I}})^{p_{I}}-\sum_{j=1}^{\ell_I-1}f_{I,j}(\u_{I,j})^{p_{I}}\right)^{\frac{p_{I,{\ell_I}}-p_{I}}{p_{I}}} \label{eq:defF}
\end{align}
and 
\begin{align*}
  \F_{I,i_{r}}(\u_{I,i_r}) &=  f_{I,i_{r}}(\u_{I,i_{r}})^{p_{I}-p_{I,i_{r}}} =
  \left(\sum_{k=1}^{\ell_{I,i_r}}f_{I,i_{r},k}(\u_{I,i_{r},k})^{p_{I,i_{r}}}\right)^{\frac{p_{I}-p_{I,i_{r}}}{p_{I,i_{r}}}}.
\end{align*}

The next two lemmata are required for the proof of Proposition
\ref{pro:DetJacobian}.  We use the somewhat sloppy notation $k \in
{I,i_r}$ if the variable $u_k$ is a leaf in the subtree below
${I,i_r}$. The same notation is used for ${\widehat{I}}$.
\begin{lemma}
  \label{lem:DerRec}Let $I=i_{1},...,i_{r-1}$ and ${I,i_{r}}$ be
  any node of the tree associated with an $L_p$-nested function
  $f$. Then the following recursions hold for the derivatives of
  $\g_{I,i_{r}}(\u_{\widehat{I,i_r}})^{p_{I,i_{r}}}$ and $f_{I,i_{r}}^{p_{I}}(\u_{I,i_{r}})$ w.r.t $u_{q}$:
  If $u_{q}$ is not in the subtree under the node $I,i_r$, i.e. $k
  \not\in I,i_r$, then
  \begin{align*}
    &\frac{\partial}{\partial u_{q}}f_{I,i_{r}}(\u_{I,i_{r}})^{p_{I}} = 0\\
    &\text{ \em and }\\
    &\frac{\partial}{\partial u_{q}}\g_{I,i_{r}}(\u_{\widehat{I,i_r}})^{p_{I,i_{r}}} 
       = \frac{p_{I,i_{r}}}{p_{I}}\G_{I,i_{r}}(\bm u_{\widehat{I,i_r}})\cdot
   \begin{cases}
     \frac{\partial}{\partial u_{q}}\g_{I}(\u_{\widehat{I}})^{p_{I}} & \mbox{ if
     }q\in  I  \\ 
     \\-\frac{\partial}{\partial u_{q}}f_{I,j}(\u_{I,j})^{p_{I}} & \mbox{ if
     }q\in I,j
   \end{cases}  
  \end{align*}
  for $q\in I,j$ and $q\not\in I,k$ for $k\not=j$.
  Otherwise 
  \begin{align*}
    \frac{\partial}{\partial u_{q}}\g_{I,i_{r}}(\u_{\widehat{I,i_r}})^{p_{I,i_{r}}} = 0 
    &\mbox{ \em and }
    \frac{\partial}{\partial u_{q}}f_{I,i_{r}}(\u_{I,i_{r}})^{p_{I}} =
    \frac{p_{I}}{p_{I,i_{r}}}\F_{I,i_{r}}(\u_{I,i_r})\frac{\partial}{\partial
      u_{q}}f_{I,i_{r},s}(\u_{I,i_{r},s})^{p_{I,i_{r}}}
  \end{align*}
  for $q\in {I,i_{r},s}$ and $q\not\in I,i_{r},k$ for $k\not=s$.
\end{lemma}
\begin{proof} Both of the first equations are obvious, since only
  those nodes have a non-zero derivative for which the subtree
  actually depends on $u_q$. The second equations can be seen by
  direct computation
 \begin{align*}
   \frac{\partial}{\partial u_{q}}\g_{I,i_{r}}(\u_{\widehat{I,i_r}})^{p_{I,i_{r}}} &= p_{I,i_{r}}\g_{I,i_{r}}(\u_{\widehat{I,i_r}})^{p_{I,i_{r}}-1}\frac{\partial}{\partial u_{q}}\G_{I,i_{r}}(\bm u_{\widehat{I,i_r}})\\
   &= p_{I,i_{r}}\g_{I,i_{r}}(\u_{\widehat{I,i_r}})^{p_{I,i_{r}}-1}\frac{\partial}{\partial u_{q}}\(\g_{I}(\u_{\widehat{I}})^{p_{I}}-\sum_{j=1}^{{\ell_I}-1}f_{I,j}(\u_{I,j})^{p_{I}}\)^{\frac{1}{p_{I}}}\\
   &= \frac{p_{I,i_{r}}}{p_{I}}\g_{I,i_{r}}(\u_{\widehat{I,i_r}})^{p_{I,i_{r}}-1}\g_{I,i_{r}}(\u_{\widehat{I,i_r}})^{1-p_{I}}\frac{\partial}{\partial u_{q}}\(\g_{I}(\u_{\widehat{I}})^{p_{I}}-\sum_{j=1}^{{\ell_I}-1}f_{I,j}(\u_{I,j})^{p_{I}}\)\\
   &= \frac{p_{I,i_{r}}}{p_{I}}\G_{I,i_{r}}(\bm u_{\widehat{I,i_r}}) \cdot
   \begin{cases}
     \frac{\partial}{\partial u_{q}}\g_{I}(\u_{\widehat{I}})^{p_{I}} & \mbox{ if
     }q\in I\\
     \\-\frac{\partial}{\partial u_{q}}f_{I,j}(\u_{I,j})^{p_{I}} & \mbox{ if
     }q\in {I,j}
   \end{cases}
\end{align*}
 
 Similarly 
\begin{align*}
  \frac{\partial}{\partial u_{q}}f_{I,i_{r}}(\u_{I,i_{r}})^{p_{I}} &= p_{I}f_{I,i_{r}}(\u_{I,i_{r}})^{p_{I}-1}\frac{\partial}{\partial u_{q}}f_{I,i_{r}}(\u_{I,i_{r}})\\
  &= p_{I}f_{I,i_{r}}(\u_{I,i_{r}})^{p_{I}-1}\frac{\partial}{\partial u_{q}}\(\sum_{k=1}^{{\ell_{I,i_r}}}f_{I,i_{r},k}(\u_{I,i_{r},k})^{p_{I,i_{r}}}\)^{\frac{1}{p_{I,i_{r}}}}\\
 &= \frac{p_{I}}{p_{I,i_{r}}}f_{I,i_{r}}(\u_{I,i_{r}})^{p_{I}-1}f_{I,i_{r}}(\u_{I,i_{r}})^{1-p_{I,i_{r}}}\frac{\partial}{\partial u_{q}}f_{I,i_{r},s}(\u_{I,i_{r},s})^{p_{I,i_{r}}}\\
 &= \frac{p_{I}}{p_{I,i_{r}}}\F_{I,i_{r}}(\u_{I,i_r})\frac{\partial}{\partial
 u_{q}}f_{I,i_{r},s}(\u_{I,i_{r},s})^{p_{I,i_{r}}}
\end{align*}
 for $k\in {I,i_{r},s}$.
\end{proof}
The next lemma states the form of the whole derivative $\frac{\partial
  u_n}{\partial u_q}$ in terms of the $\G_I(\bm u_{\widehat I})$- and $\F_I(\u_{I})$-terms.
\begin{lemma}
  \label{lem:derivatives}Let
  $|u_{q}|=\f_{\ell_{1},...,\ell_{m},i_{1},...,i_{t}}$,
  $|u_{n}|=\f_{\ell_{1},...,\ell_{d}}$ with $m<d$.  The derivative of
  $u_{n}$ w.r.t. $u_{q}$ is given by
\begin{align*}
  \frac{\partial}{\partial u_{q}}u_{n} &= -\G_{\ell_{1},...,\ell_{d}}(\bm u_{\widehat{\ell_{1},...,\ell_{d}}})\cdot...\cdot \G_{\ell_{1},...,\ell_{m+1}}(\bm u_{\widehat{\ell_{1},...,\ell_{m+1}}})\\
  &\times \F_{\ell_{1},...,\ell_{m},i_{1}}(\u_{\ell_{1},...,\ell_{m},i_{1}})\cdot \F_{\ell_{1},...,\ell_{m},i_{1},...,i_{t-1}}(\u_{\ell_{1},...,\ell_{m},i_{1},...,i_{t-1}})\cdot\Delta_{q}|u_{q}|^{p_{\ell_{1},...,\ell_{m},i_{1},...,i_{t-1}}-1}\end{align*}
with $\Delta_{q}=\mathrm{sgn}\, u_{q}$ and $|u_{q}|^{p}=(\Delta_{q}u_{q})^{p}$.
In particular
\begin{align*}
  u_{q}\frac{\partial}{\partial u_{q}}u_{n} &=
  -\G_{\ell_{1},...,\ell_{d}}(\bm
  u_{\widehat{\ell_{1},...,\ell_{d}}})\cdot...\cdot
  \G_{\ell_{1},...,\ell_{m+1}}(\bm
  u_{\widehat{\ell_{1},...,\ell_{m+1}}})\\
  &\times
  \F_{\ell_{1},...,\ell_{m},i_{1}}(\u_{1})\cdot
  \F_{\ell_{1},...,\ell_{m},i_{1},...,i_{t-1}}(\u_{\ell_{1},...,\ell_{m},i_{1}})\cdot|u_{q}|^{p_{\ell_{1},...,\ell_{m},i_{1},...,i_{t-1}}}.\end{align*}
\end{lemma}
\begin{proof}
Successive application of Lemma (\ref{lem:DerRec}).
 \end{proof}

\begin{proof}[Proposition \ref{pro:DetJacobian}]
  Before we begin with the proof, note that $\F_I(\u_{I})$ and $\G_I(\bm u_{\widehat I})$ fulfill
  following equalities
  \begin{eqnarray}
    \G_{I,i_{m}}(\bm u_{\widehat{I,i_m}})^{-1}\g_{I,i_{m}}(\u_{\widehat{I,i_m}})^{p_{I,i_{m}}} & = & \g_{I,i_{m}}(\u_{\widehat{I,i_m}})^{p_{I}}\label{eq:zt1}\\
    & = &
    \g_{I}(\u_{\widehat{I}})^{p_{I}}-\sum_{k=1}^{\ell_I-1}\F_{I,k}(\u_{I,k})f_{I,k}(\u_{I,k})^{p_{I,k}}\label{eq:zt2}
  \end{eqnarray}
  and
  \begin{eqnarray}
    f_{I,i_{m}}(\u_{I,i_{m}})^{p_{I,i_{m}}} & = &
    \sum_{k=1}^{\ell_{I,i_{m}}}\F_{I,i_{m},k}(\u_{I,i_m,k})f_{I,i_{m},k}(\u_{I,i_{m},k})^{p_{I,i_{m},k}}.\label{eq:zt3}
  \end{eqnarray}

  Now let $L=\ell_{1},...,\ell_{d-1}$ be the multi-index of the parent
  of $u_n$. We compute
  $\frac{1}{r^{n-1}}|\det\mathcal{J}|$ and obtain the result by
  solving for $|\det\mathcal{J}|$. As shown in Lemma
  (\ref{lem:generalDeterminant}) $\frac{1}{r^{n-1}}|\det\mathcal{J}|$ has
  the form\begin{eqnarray*} \frac{1}{r^{n-1}}|\det\mathcal{J}| & = &
    -\sum_{k=1}^{n-1}\frac{\partial u_{n}}{\partial u_{k}}\cdot
    u_{k}+u_{n}.\end{eqnarray*} By definition
  $u_{n}=\g_{L,\ell_{d}}(\u_{\widehat{L,\ell_d}})=\g_{L,\ell_{d}}(\u_{\widehat{L,\ell_d}})^{p_{L,\ell_{d}}}$.  Now, assume
  that $u_{m},...,u_{n-1}$ are children of $L$, i.e.
  $u_{k}=\f_{L,I,i_{t}}$ for some $I,i_{t}=i_{1},...,i_{t}$ and $m\le k<n$. Remember,
  that by Lemma (\ref{lem:derivatives}) the terms
  $u_{q}\frac{\partial}{\partial u_{q}}u_{n}$ for $m\le q< n$ have
  the form
\begin{align*} u_{q}\frac{\partial}{\partial u_{q}}u_{n}
    &= -\G_{L,\ell_{d}}(\bm u_{\widehat{L,\ell_d}})\cdot \F_{L,i_{1}}(\u_{L,i_1})\cdot...\cdot
    \F_{L,I}(\u_{L,I})\cdot|u_{q}|^{p_{\ell_{1},...,\ell_{d-1},i_{1},...,i_{t-1}}}.\end{align*}
  Using equation (\ref{eq:zt2}), we can expand the determinant as follows \begin{align*}
    &   -\sum_{k=1}^{n-1}\frac{\partial u_{n}}{\partial u_{k}}\cdot u_{k}+\g_{L,\ell_{d}}(\u_{\widehat{L,\ell_d}})^{p_{L,\ell_{d}}}\\
    &= -\sum_{k=1}^{m-1}\frac{\partial u_{n}}{\partial u_{k}}\cdot u_{k}-\sum_{k=m}^{n-1}\frac{\partial u_{n}}{\partial u_{k}}\cdot u_{k}+\g_{L,\ell_{d}}(\u_{\widehat{L,\ell_d}})^{p_{L,\ell_{d}}}\\
    &= -\sum_{k=1}^{m-1}\frac{\partial u_{n}}{\partial u_{k}}\cdot
    u_{k}+\G_{L,\ell_{d}}(\bm u_{\widehat{L,\ell_d}})\left(-\sum_{k=m}^{n-1}\G_{L,\ell_{d}}(\bm u_{\widehat{L,\ell_d}})^{-1}\frac{\partial
        u_{n}}{\partial u_{k}}\cdot
      u_{k}+\G_{L,\ell_{d}}(\bm u_{\widehat{L,\ell_d}})^{-1}\g_{L,\ell_{d}}(\u_{\widehat{L,\ell_d}})^{p_{L,\ell_{d}}}\right)\\
    &= -\sum_{k=1}^{m-1}\frac{\partial u_{n}}{\partial u_{k}}\cdot u_{k}+\G_{L,\ell_{d}}(\bm u_{\widehat{L,\ell_d}})\left(-\sum_{k=m}^{n-1}\G_{L,\ell_{d}}(\bm u_{\widehat{L,\ell_d}})^{-1}\frac{\partial u_{n}}{\partial u_{k}}\cdot u_{k}+\g_{L}(\u_{\widehat{L}})^{p_{L}}-\sum_{k=1}^{{\ell_d}-1}\F_{L,k}(\u_{L,k})f_{L,k}(\u_{L,k})^{p_{L,k}}\right).
\end{align*}
Note that all terms
$\G_{L,\ell_{d}}(\bm u_{\widehat{L,\ell_d}})^{-1}\frac{\partial
  u_{n}}{\partial u_{k}}\cdot u_{k}$ for $m\le k < n$ now have the
form
\begin{align*}
  \G_{L,\ell_{d}}(\bm u_{\widehat{L,\ell_d}})^{-1}u_{k}\frac{\partial}{\partial u_{k}}u_{n} & = 
  -\F_{L,i_{1}}(\u_{L,i_1})\cdot...\cdot
  \F_{L,I}(\u_{L,I})\cdot|u_{q}|^{p_{\ell_{1},...,\ell_{d-1},i_{1},...,i_{t-1}}}
\end{align*}
since we constructed them to be neighbors of $u_{n}$. However, with
equation (\ref{eq:zt3}), we can further
expand the sum $\sum_{k=1}^{{\ell_d}-1}\F_{L,k}(\u_{L,k})f_{L,k}(\u_{L,k})^{p_{L,k}}$ down to
the leaves $u_{m},...,u_{n-1}$. When doing so we end up with
the same factors $\F_{L,i_{1}}(\u_{L,i_1})\cdot...\cdot
\F_{L,I}(\u_{L,I})\cdot|u_{q}|^{p_{\ell_{1},...,\ell_{d-1},i_{1},...,i_{t-1}}}$
as in the derivatives
$\G_{L,\ell_{d}}(\bm u_{\widehat{L,\ell_d}})^{-1}u_{q}\frac{\partial}{\partial u_{q}}u_{n}$.
This means exactly that 
$$-\sum_{k=m}^{n-1}\G_{L,\ell_{d}}(\bm u_{\widehat{L,\ell_d}})^{-1}\frac{\partial u_{n}}{\partial
  u_{k}}\cdot u_{k} = \sum_{k=1}^{{\ell_d}-1}\F_{L,k}(\u_{L,k})f_{L,k}(\u_{L,k})^{p_{L,k}}$$ 
and, therefore,
\begin{align*}
    &= -\sum_{k=1}^{m-1}\frac{\partial u_{n}}{\partial u_{k}}\cdot u_{k}+\G_{L,\ell_{d}}(\bm u_{\widehat{L,\ell_d}})\left(-\sum_{k=m}^{n-1}\G_{L,\ell_{d}}(\bm u_{\widehat{L,\ell_d}})^{-1}\frac{\partial u_{n}}{\partial u_{k}}\cdot u_{k}+\g_{L}(\u_{\widehat{L}})^{p_{L}}-\sum_{k=1}^{{\ell_d}-1}\F_{L,k}(\u_{L,k})f_{L,k}(\u_{L,k})^{p_{L,k}}\right)\\
    &= -\sum_{k=1}^{m-1}\frac{\partial u_{n}}{\partial u_{k}}\cdot u_{k}+\G_{L,\ell_{d}}(\bm u_{\widehat{L,\ell_d}})\left(\sum_{k=1}^{{\ell_d}-1}\F_{L,k}(\u_{L,k})f_{L,k}(\u_{L,k})^{p_{L,k}}+\g_{L}(\u_{\widehat{L}})^{p_{L}}-\sum_{k=1}^{{\ell_d}-1}\F_{L,k}(\u_{L,k})f_{L,k}(\u_{L,k})^{p_{L,k}}\right)\\
    &= -\sum_{k=1}^{m-1}\frac{\partial u_{n}}{\partial u_{k}}\cdot u_{k}+\G_{L,\ell_{d}}(\bm u_{\widehat{L,\ell_d}})\g_{L}(\u_{\widehat{L}})^{p_{L}}.
\end{align*}

By factoring out $\G_{L,\ell_{d}}(\bm u_{\widehat{L,\ell_d}})$ from the equation, the terms
$\frac{\partial u_{n}}{\partial u_{k}}\cdot u_{k}$ loose the
$\G_{L,\ell_{d}}$ in front and we get basically the same equation as
before, only that the new leaf (the new ``$u_{n}$'') is $\g_{L}(\u_{\widehat{L}})^{p_L}$
and we got rid of all the children of ${L}$. By repeating that
procedure up to the root node, we successively factor out all
$\G_{L'}(\bm u_{\widehat{L'}})$ for $L'\in\mathcal L$ until all terms of the sum vanish and
we are only left with $\f_{\emptyset}=1$.  Therefore, the determinant
is
\begin{align*}
  \frac{1}{r^{n-1}}|\det\mathcal{J}| & = 
  \prod_{L\in\mathcal L}\G_{L}(\bm u_{\widehat L})
\end{align*} which
completes the proof.\end{proof}

\section{Volume and Surface of the $L_p$-Nested Unit Sphere}
\label{app:VolumeSurface}
\begin{proof}[Proposition \ref{pro:VolumeSurface}]
  We obtain the volume by computing the integral $\int_{f(\x)\le
    R}d\x$. Differentiation with respect to $R$ yields the surface
  area. For symmetry reasons we can compute the volume only on the
  positive quadrant $\R_+^n$ and multiply the result with $2^n$ later
  to obtain the full volume and surface area.  The strategy for
  computing the volume is as follows. We start off with inner nodes
  $I$ that are parents of leaves only. The value $\f_I$ of such a node
  is simply the $L_{p_I}$ norm of its children. Therefore, we can
  convert the integral over the children of $I$ with the
  transformation of \cite{gupta:1997}. This maps the leaves
  $\ff_{I,1:\ell_I}$ into $\f_I$ and ``angular'' variables
  $\twbu$. Since integral borders of the original integral
  depend only on the value of $\f_I$ and not on $\twbu$, we can
  separate the variables $\twbu$ from the radial variables $\f_I$ and
  integrate the variables $\twbu$ separately. The
  integration over $\twbu$ yields a certain factor, while
  the variable $\f_I$ effectively becomes a new leaf.

  Now suppose $I$ is the parent of leaves only. Without loss of
  generality let the $\ell_I$ leaves correspond to the last $\ell_I$
  coefficients of $\x$. Let $\x\in \R_+^n$. Carrying out the first
  transformation and integration yields
  \begin{align*}
    \int_{f(\x)\le R} d\x &= \int_{f(\x_{1:n-\ell_I, \f_I})\le R}
    \int_{\twbu\in \mathcal V_+^{\ell_I-1}}
    \f_{I}^{\ell_{I}-1}\left(1-\sum_{i=1}^{\ell_{I}-1}\twu_{i}^{p_{I}}\right)^{\frac{1-p_{I}}{p_{I}}}d\f_{I}d\twbu
    d\x_{1:n-\ell_I}\\
    &=
    \int_{f(\x_{1:n-\ell_I, \f_I})\le R}\f_{I}^{n_{I}-1}d\f_{I}d\x_{1:n-\ell_I}\times\int_{\twbu\in \mathcal V_+^{\ell_I-1}}\left(1-\sum_{i=1}^{\ell_{I}-1}\twu_{i}^{p_{I}}\right)^{\frac{n_{I,\ell_{I}}-p_{I}}{p_{I}}}d\twbu.
 \end{align*}
 For solving the second integral we make the pointwise transformation
 $s_i = \twu_i^{p_I}$ and obtain
\begin{align*}
\int_{\twbu\in \mathcal
  V_+^{\ell_I-1}}\left(1-\sum_{i=1}^{\ell_{I}-1}\twu_{i}^{p_{I}}\right)^{\frac{n_{I,\ell_{I}}-p_{I}}{p_{I}}}d\twbu
&= \frac{1}{p_I^{\ell_I-1}}\int_{\sum s_i \le
  1}\left(1-\sum_{i=1}^{\ell_{I}-1}s_i\right)^{\frac{n_{I,\ell_{I}}}{p_{I}}-1}
\prod_{i=1}^{\ell_I-1} s_i^{\frac{1}{p_I}-1}d{\bm
  s_{\ell_I-1}}
\\
&= \frac{1}{p_I^{\ell_I-1}}\prod_{k=1}^{\ell_I-1}
B\left[\frac{\sum_{i=1}^{k}
    n_{I,k}}{p_I},\frac{n_{I,k+1}}{p_I}\right]\\
&= \frac{1}{p_I^{\ell_I-1}}\prod_{k=1}^{\ell_I-1}
B\left[\frac{k}{p_I},\frac{1}{p_I}\right]
\end{align*}
by using the fact that the transformed integral has the form of an
unnormalized Dirichlet distribution and, therefore, the value of the
integral must equal its normalization constant.

Now, we solve the integral 
\begin{align}
  \int_{f(\x_{1:n-\ell_I, \f_I})\le
    R}\f_{I}^{n_{I}-1}d\f_{I}d\x_{1:n-\ell_I}.\label{eq:leftIntegralStart}
\end{align}
We carry this out in exactly the same manner as we solved the previous
integral. We only need to make sure that we only contract nodes that
have only leaves as children (remember that radii of contracted nodes
become leaves) and we need to find a formula how the factors
$\f_I^{n_I-1}$ propagate through the tree.

For the latter, we first state the formula and then prove it via
induction. For notational convenience let $\mathcal J$ denote the set
of multi-indices corresponding to the contracted leaves,
$\x_{\widehat{\mathcal J}}$ the remaining coefficients of $\x$ and
$\ff_{\mathcal J}$ the vector of leaves resulting from
contraction. The integral which is left to solve after integrating
over all $\twbu$ is given by (remember that $n_{J}$ denotes real
leaves, i.e. the ones corresponding to coefficients of $\x$):
\begin{align*}
  \int_{f(\x_{\widehat{\mathcal J}},\ff_{\mathcal J})\le R}\prod_{J\in\mathcal J}\f_{J}^{n_J-1}d\ff_{\mathcal J} d\x_{\widehat{\mathcal J}}.
\end{align*}
We already proved the first induction step by computing equation
(\ref{eq:leftIntegralStart}). For computing the general induction step
suppose $I$ is an inner node whose children are leaves or contracted
leaves. Let $\mathcal J'$ be the set of contracted leaves under $I$
and ${\mathcal K}=\mathcal J \backslash \mathcal J'$.
Transforming the children of $I$ into radial coordinates by
\cite{gupta:1997} yields
\begin{align*}
  \int_{f(\x_{\widehat{\mathcal J}},\ff_{\mathcal J})\le R}\prod_{J\in\mathcal
    J}\f_{J}^{n_J-1}d\ff_{\mathcal J} d\x_{\widehat{\mathcal J}}
&=   \int_{f(\x_{\widehat{\mathcal J}},\ff_{\mathcal J})\le R}\(\prod_{K\in{\mathcal K}}
\f_{K}^{n_{K}-1}\)\cdot \(\prod_{J'\in\mathcal
    J'}\f_{J'}^{n_{J'}-1}\) d\ff_{\mathcal J} d\x_{\widehat{\mathcal J}}\\
&=    \int_{f(\x_{\widehat{{\mathcal K}}},\ff_{\mathcal K},\f_I)\le R}\int_{\twbu_{\ell_I-1}\in \mathcal
  V_+^{\ell_I-1}}
\( \(1-\sum_{i=1}^{\ell_I-1}\twu_i^{p_I}\)^{\frac{1-p_I}{p_I}}
\f_I^{\ell_I-1}\)\cdot\( \prod_{K\in{\mathcal K}} \f_{K}^{n_{K}-1} \)\\
  &\times \(\(\f_I
  \(1-\sum_{i=1}^{\ell_I-1}\twu_i^{p_I}\)\)^{\frac{n_{\ell_I}-1}{p_I}}\prod_{k=1}^{\ell_I-1}\(\f_I
  \twu_k\)^{n_k-1}\) d\x_{\widehat{{\mathcal K}}} d\ff_{{\mathcal K}} d\f_I d\twbu_{\ell_I-1}\\
&=    \int_{f(\x_{\widehat{{\mathcal K}}},\ff_{\mathcal K},\f_I)\le R}\int_{\twbu_{\ell_I-1}\in \mathcal
  V_+^{\ell_I-1}}
\( \prod_{K\in{\mathcal K}} \f_{K}^{n_{K}-1} \)\\
  &\times \(\f_I^{\ell_I-1+\sum_{i=1}^{\ell_I}(n_i-1)}
  \(1-\sum_{i=1}^{\ell_I-1}\twu_i^{p_I}\)^{\frac{n_{\ell_I}-p_I}{p_I}}\prod_{k=1}^{\ell_I-1}
  \twu_k^{n_k-1}\) d\x_{\widehat{{\mathcal K}}} d\ff_{{\mathcal K}} d\f_I d\twbu_{\ell_I-1}\\
&=    \int_{f(\x_{\widehat{{\mathcal K}}},\ff_{\mathcal K},\f_I)\le R}
\( \prod_{K\in{\mathcal K}} \f_{K}^{n_{K}-1} \)\f_I^{n_I-1} d\x_{\widehat{{\mathcal K}}} d\ff_{{\mathcal K}} d\f_I \\
  &\times \int_{\twbu_{\ell_I-1}\in \mathcal
  V_+^{\ell_I-1}}
  \(1-\sum_{i=1}^{\ell_I-1}\twu_i^{p_I}\)^{\frac{n_{\ell_I}-p_I}{p_I}}\prod_{k=1}^{\ell_I-1}
  \twu_k^{n_k-1}  d\twbu_{\ell_I-1}.
\end{align*}
Again, by transforming it into a Dirichlet distribution, the latter
integral has the solution
$$\int_{\twbu_{\ell_I-1}\in \mathcal
  V_+^{\ell_I-1}}
  \(1-\sum_{i=1}^{\ell_I-1}\twu_i^{p_I}\)^{\frac{n_{\ell_I}-p_I}{p_I}}\prod_{k=1}^{\ell_I-1}
  \twu_k^{n_k-1}  d\twbu_{\ell_I-1} = \prod_{k=1}^{\ell_I-1}
B\left[\frac{\sum_{i=1}^{k}
    n_{I,k}}{p_I},\frac{n_{I,k+1}}{p_I}\right]$$
while the remaining former integral has the form 
$$\int_{f(\x_{\widehat{{\mathcal K}}},\ff_{\mathcal K},\f_I)\le R}
\( \prod_{K\in{\mathcal K}} \f_{K}^{n_{K}-1} \)\f_I^{n_I-1} d\x_{\widehat{{\mathcal K}}} d\ff_{{\mathcal K}} d\f_I =
\int_{f(\x_{\widehat{\mathcal J}},\ff_{\mathcal J})\le R}\prod_{J\in\mathcal
  J}\f_{J}^{n_J-1}d\ff_{\mathcal J} d\x_{\widehat{\mathcal J}}$$
as claimed. 

By carrying out the integration up to the root node the remaining
integral becomes
$$\int_{\f_\nod \le R} \f_\nod^{n-1}d\f_\nod=\int_0^R
\f_\nod^{n-1}d\f_\nod=\frac{R^n}{n}.$$ Collecting the factors from
integration over the $\twbu$ proves the equations (\ref{eq:volbeta})
and (\ref{eq:surbeta}). Using
$B\left[a,b\right]=\frac{\Gamma[a]\Gamma[b]}{\Gamma[a+b]}$ yields
equations (\ref{eq:volgamma}) and (\ref{eq:surgamma}).
\end{proof}

\section{Layer Marginals}
\label{app:layermarginals}
\begin{proof}[Proposition \ref{pro:layermarginals}]
  \begin{align*}
    \rho(\x) &= \frac{\varrho(f(\x))}{\mathcal  S_f(f(\x)) } \\
    &= \frac{\varrho(f(\x_{1:n-\ell_I},\f_I,\twbu_{\ell_I-1},\Delta_{n}))}{\mathcal
      S_f(f(\x)) }\cdot 
      \f_{I}^{\ell_{I}-1}\left(1-\sum_{i=1}^{\ell_{I}-1}|\twu_{i}|^{p_{I}}\right)^{\frac{1-p_{I}}{p_{I}}}
 \end{align*}
 where $\Delta_n = \text{sign} (x_n)$. Note that $f$ is invariant to
 the actual value of $\Delta_n$. However, when integrating it out, it
 yields a factor of $2$. Integrating out $\twbu_{\ell_I-1}$ and
 $\Delta_n$ now yields 
  \begin{align*}
    \rho(\x_{1:n-\ell_I},\f_I)&= \frac{\varrho(f(\x_{1:n-\ell_I},\f_I))}{\mathcal
      S_f(f(\x)) }\cdot 
      \f_{I}^{\ell_{I}-1} \frac{2^{\ell_I}
        \Gamma^{\ell_I}\left[\frac{1}{p_I}\right]}{p_I^{\ell_I-1}\Gamma\left[\frac{\ell_I}{p_I}\right]}\\
      &= \frac{\varrho(f(\x_{1:n-\ell_I},\f_I))}{\mathcal
      S_f(f(\x_{1:n-\ell_I},\f_I)) }\cdot 
      \f_{I}^{\ell_{I}-1} \\
 \end{align*}
 Now, we can go on an integrate out more subtrees. For that purpose,
 let $\x_{\widehat{\mathcal J}}$ denote the remaining coefficients of $\x$, $\ff_{\mathcal J}$ the
 vector of leaves resulting from the kind of contraction just shown for
 $\f_I$ and $\mathcal J$ the set of multi-indices corresponding to the
 ``new leaves'', i.e the node $\f_I$ after contraction. We obtain the
 following equation
 \begin{align*}
   \rho( \x_{\widehat{\mathcal J}},\ff_{\mathcal J}) &= \frac{\varrho(f(\x_{\widehat{\mathcal J}},\ff_{\mathcal J}))}{S_f(f(\x_{\widehat{\mathcal J}},\ff_{\mathcal J}))}\prod_{J\in\mathcal J}\f_{J}^{n_J-1}.
 \end{align*}
 where $n_J$ denotes the number of leaves in the subtree under the
 node $J$. The calculations for the proof are basically the same as
 the one for proposition (\ref{pro:VolumeSurface}).
\end{proof}

\section{Factorial $L_p$-Nested Distributions}
\label{app:noIndep}
\begin{proof}[Proposition \ref{pro:noIndep}]
  Since the single $x_i$ are independent,
  $f_1(\x_{1}),...,f_{\ell_\nod}(\x_{{\ell_\nod}})$ and, therefore,
  $\f_1,...,\f_{\ell_\nod}$ must be independent as well ($\x_i$ are
  the elements of $\x$ in the subtree below the $i$th child of the
  root node). Using Corollary \ref{cor:layermarginals} we can write
  the density of $\f_1,...,\f_{\ell_\nod}$ as (the function name $g$
  is unrelated to the usage of the function $g$ above)

\begin{align*}
  \rho(\ff_{1:\ell_\nod}) &= \prod_{i=1}^{\ell_\nod} h_i(\f_i) =
  g(\|\ff_{1:\ell_\nod}\|_{p_\nod})
  \prod_{i=1}^{\ell_\nod}\f_i^{n_i-1}
\end{align*}
with
\begin{align*}
g(\|\ff_{1:\ell_\nod}\|_{p_\nod}) &= \frac{p_\nod^{\ell_\nod -1} \Gamma\left[\frac{n}{p_\nod}\right]}
    {f(\f_1,...,\f_{\ell_\nod})^{n -1} 2^m \prod_{k=1}^{\ell_\nod}\Gamma\left[\frac{n_{k}}{p_\nod}\right]}
    \varrho\(\|\ff_{1:\ell_\nod}\|_{p_\nod}\)
\end{align*}
Since the integral over $g$ is finite, it follows from
\cite{sinz:2009a} that $g$ has the form
$g(\|\ff_{1:\ell_{\nod}}\|_{p_\nod}) = \exp(a_\nod
\|\ff_{1:\ell_{\nod}}\|_{p_\nod}^{p_\nod} + b_\nod)$ for
appropriate constants $a_\nod$ and $b_\nod$. Therefore, the marginals
have the form 
\begin{align}
h_i(\f_i) = \exp(a_\nod \f_i^{p_\nod} + c_\nod)
\f_i^{n_i-1}\label{eq:marginalform}.
\end{align}

On the other hand, the particular form of $g$ implies that the radial
density has the form $\varrho(f(\x)) \propto f(\x)^{(n-1)} \exp(a_\nod
f(\x)^{p_\nod} + b_\nod)^{p_\nod}$. In particular, this implies that the root
node's children $f_i(\x_i)$ ($i=1,...,\ell_\nod$) are independent and
$L_p$-nested again. With the same argument as above, it follows that
their children $\ff_{i,1:\ell_i}$ follow the distribution
$\rho(\f_{i,1},...,\f_{i,\ell_i})= \exp(a_i
\|\ff_{i,1:\ell_i}\|_{p_i}^{p_i} + b_i) \prod_{j=1}^{\ell_i}
\f_{i,j}^{n_{i,j}-1}$. Transforming that distribution to
$L_p$-spherically symmetric polar coordinates
$\f_i=\|\ff_{i,1:\ell_i}\|_{p_i}$ and $\twbu = \ff_{i,1:\ell_i-1} /
\|\ff_{i,1:\ell_i}\|_{p_i}$ as in \cite{gupta:1997}, we obtain the
form
\begin{align*}
\rho(\f_{i},\twbu)&= \exp(a_i \f_{i}^{p_i} +
b_i)\f_i^{\ell_i - 1} \(1-\sum_{j=1}^{\ell_i-1}
|\twu_i|^{p_i}\)^\frac{1-p_i}{p_i} 
\(\f_i \(1-\sum_{j=1}^{\ell_i-1}
|\twu_i|^{p_i}\)^\frac{1}{p_i} \)^{n_{i,\ell_i}-1}  \prod_{j=1}^{\ell_i-1}
(\twu_j \f_i)^{n_{i,j}-1}\\
&= \exp(a_i \f_{i}^{p_i} +
b_i) \f_i^{n_i-1} \(1-\sum_{j=1}^{\ell_i-1}
|\twu_i|^{p_i}\)^\frac{n_{i,\ell_i}-p_i}{p_i} 
  \prod_{j=1}^{\ell_i-1}
\twu_j^{n_{i,j}-1},
\end{align*}
where the second equation follows the same calculations as in the
proof of \ref{pro:VolumeSurface}.  After integrating out $\twbu$,
assuming that the $x_i$ are statistically independent, we obtain the
density of $\f_i$ which is equal to (\ref{eq:marginalform}) if and
only if $p_i=p_\nod$. However, if $p_\nod$ and $p_i$ are equal, the
hierarchy of the $L_p$-nested function shrinks by one layer since
$p_i$ and $p_\nod$ cancel themselves. Repeated application of the
above argument collapses the complete $L_p$-nested tree until one
effectively obtains an $L_p$-spherical function. Since the only
factorial $L_p$-spherically symmetric distribution is the
$p$-generalized Normal \citep{sinz:2009a} the claim follows.
\end{proof}

\section{Determinant of the Jacobian for NRF}
\label{app:radialJacobian}
\begin{proof}[Lemma \ref{lem:radialDeterminant}]
  The proof is a generalization of the proof of \cite{lyu:2009}. Due
  to the chain rule the Jacobian of the entire transformation is the
  multiplication of the Jacobians for each single step, i.e. the
  rescaling of a subset of the dimensions for one single inner
  node. The Jacobian for the other dimensions is simply the identity
  matrix. Therefore, the determinant of the Jacobian for each single
  step is the determinant for the radial transformation on the
  respective dimensions. We show how to compute the determinant for a
  single step. 

  Assume that we reached a particular node $I$ in Algorithm
  \ref{alg:non-linearICA}. The leaves, which have been rescaled by the
  preceding steps, are called $\t_I$. Let ${\bm \xi}_I =
  \frac{g_I(f_I(\t_I))}{f_I(\t_I))} \cdot \t_I$ with $g_I(r) =
  (\mathcal F_{\bot\!\!\!\bot}^{-1}\circ \mathcal F_{s})(r)$. The
  general form of a single Jacobian is
  \begin{align*}
    \frac{\partial{\bm \xi}_I}{\partial\t_I}&=\t_I\cdot\frac{\partial}{\partial\t_I}\left(\frac{g_I(f_I(\t_I))}{f_I(\t_I)}\right)+\frac{g_I(f_I(\t_I))}{f_I(\t_I)}I_{n_I},
  \end{align*}
  where
  \begin{align*}
    \frac{\partial}{\partial\t_I}\left(\frac{g_I(f_I(\t_I))}{f_I(\t_I)}\right)&=\(\frac{g_I'(f_I(\t_I))}{f_I(\t_I)}-\frac{g_I(f_I(\t_I))}{f_I(\t_I)^{2}}\)\frac{\partial}{\partial\t_I}f_I(\t_I).
  \end{align*}

  Let $y_{i}$ be a leave in the subtree under $I$ and let $I,
  J_1,...,J_k$ be the path of inner nodes from $I$ to $y_i$, then
\begin{align*}
\frac{\partial}{\partial
  y_{i}}f_I(\t_I)&=v_{I}^{1-p_{I}}v_{J_{1}}^{p_{I}-p_{J_{1}}}\cdot...\cdot
v_{k}^{p_{J_{k-1}}-p_{J_{k}}}|y_{i}|^{p_{J_{k}}-1} \cdot \mbox{sgn} y_{i}.
\end{align*}

If we denote $r=f_I(\t_I)$ and
$\zeta_{i}=v_{J_{1}}^{p_{I}-p_{J_{1}}}\cdot...\cdot
v_{k}^{p_{J_{k-1}}-p_{J_{k}}}|y_{i}|^{p_{J_{k}}-1} \cdot \mbox{sgn} y_{i}$ for the respective $J_k$, we obtain
\begin{align*}
\det\left(\t_I\cdot\frac{\partial}{\partial\t_I}\left(\frac{g_I(f_I(\t_I))}{f_I(\t_I)}\right)+\frac{g_I(f_I(\t_I))}{f_I(\t_I)} I_{n_I}\right)&=\det\left(\left(g_I'(r)-\frac{g_I(r)}{r}\right)r^{-p_{I}}\t_I\cdot{\bm \zeta}^{\top}+\frac{g_I(r)}{r} I_{n_I}\right).
\end{align*}

Now we can use Sylvester's determinant formula $\det(I_{n}+b\t_I{\bm
  \zeta}^{\top})=\det(1+b\t_I^{\top}{\bm \zeta})=1+b\t_I^{\top}{\bm
  \zeta}$ 
or equivalently
\begin{align*}
\det(aI_{n}+b\t_I{\bm \zeta}^{\top})&=\det\left(a\cdot\left(I_{n}+\frac{b}{a}\t_I{\bm \zeta}^{\top}\right)\right)\\&=a^{n}\det\left(I_{n}+\frac{b}{a}\t_I{\bm \zeta}^{\top}\right)\\&=a^{n-1}(a+b\t_I^{\top}{\bm \zeta}),
\end{align*}
as well as $\t_I^{\top}{\bm \zeta}=f_I(\t_I)^{p_{I}} = r^{p_I}$ to see that
\begin{align*}
\det\left(\left(g_I'(r)-\frac{g_I(r)}{r}\right)r^{-p_{I}}\t_I\cdot{\bm \zeta}^{\top}+\frac{g_I(r)}{r}I_{n}\right)&=\frac{g_I(r)^{n-1}}{r^{n-1}}\det\left(\left(g_I'(r)-\frac{g_I(r)}{r}\right)r^{-p_{I}}\t_I^{\top}\cdot{\bm \zeta}+\frac{g_I(r)}{r}\right)\\&=\frac{g_I(r)^{n-1}}{r^{n-1}}\det\left(g_I'(r)-\frac{g_I(r)}{r}+\frac{g_I(r)}{r}\right)\\&=\frac{g_I(r)^{n-1}}{r^{n-1}}\frac{d}{dr}g_I(r).
\end{align*}
$\frac{d}{dr}g_I(r)$ is readily computed via $\frac{d}{dr}g_I(r) = \frac{d}{dr}(\mathcal
  F_{\bot\!\!\!\bot}^{-1}\circ \mathcal F_{s})(r) =
  \frac{\varrho_s(r)}{\varrho_{\bot\!\!\!\bot}(g_I(r))}$. 

Multiplying the single determinants along with $\det W$ for the final
step of the chain rule completes the proof.
\end{proof}
\bibliographystyle{abbrvnat.bst}
\bibliography{SinzBethge2009a} 


\end{document}